\documentclass[]{siamart250211}

\usepackage{amsmath,amssymb,bm}
\usepackage{array}
\usepackage{booktabs}
\usepackage{graphicx}
\usepackage{tikz}
\usetikzlibrary{arrows.meta,positioning,fit,backgrounds}
\usepackage{algorithm}
\usepackage{algpseudocode}
\usepackage{hyperref}
\usepackage{multirow}
\usepackage{enumitem}
\usepackage{float}
\usepackage{placeins}

\newsiamremark{remark}{Remark}
\newsiamremark{heuristic}{Heuristic}
\newsiamthm{property}{Property}

\headers{BINN for Bosonic Ground States}{Z.~Kong, J.~Z.~Yang, C.~Yuan, and X.~Zhao}

\title{A Boson-informed neural network for Ground States of
       Many-Body Schr\"{o}dinger equations %
       \thanks{Submitted to the editors \today.
       \funding{J.~Z.~Yang is supported by NSFC (No.~12125103,
       No.~U24A2002) and Hubei Natural Science Foundation
       (No.~2025AFA002); X.~Zhao is supported by the National Key
       Research and Development Program of China
       (No.~2024YFE03240400), the National MCF Energy R\&D Program,
       and NSFC (No.~42450275, No.~12271413); C.~Yuan is supported by
       NSFC (No.~12301558).}}}

\author{Zhizhong Kong\thanks{School of Mathematics and Statistics,
  Wuhan University, Wuhan, 430072, Hubei, China
  (\email{zhizhongkong@whu.edu.cn}).}
\and
Jerry Zhijian Yang\thanks{School of Mathematics and Statistics,
  Wuhan University, Wuhan, 430072, Hubei, China; and
  Computational Sciences Hubei Key Laboratory, Wuhan University,
  Wuhan, 430072, Hubei, China
  (\email{zjyang.math@whu.edu.cn}).}
\and
Cheng Yuan\thanks{School of Artificial Intelligence, Wuhan University, Wuhan, 430072, Hubei, China; and
  Computational Sciences Hubei Key Laboratory, Wuhan University,
  Wuhan, 430072, Hubei, China
  (\email{yuancheng@whu.edu.cn}).}
\and
Xiaofei Zhao\thanks{School of Mathematics and Statistics,
  Wuhan University, Wuhan, 430072, Hubei, China; and
  Computational Sciences Hubei Key Laboratory, Wuhan University,
  Wuhan, 430072, Hubei, China
  (\email{matzhxf@whu.edu.cn}).}}

\begin{document}
\maketitle

\begin{abstract}
We propose a symmetry-preserving Boson-informed neural network (BINN) for the variational Monte Carlo computation of interacting-boson ground states in continuous space. The framework explicitly separates the two dominant physical scales of the system. At the microscopic interaction scale, a cell Jastrow factor absorbs the leading near-contact stiffness, while at the macroscopic scale, a potential-aware envelope encodes the external confinement, leaving the neural network to capture a smoother many-body residual. Across both two- and three-dimensional tests, the BINN approaches the Hartree and Gross-Pitaevskii limits with an empirical $\mathcal{O}(N^{-1})$ finite-size trend in both energy and density. The energies agree with diffusion Monte Carlo to within $0.06\%$ for the particle number tested up to $350$ in the three-dimensional shrinking-range case. Applications are made under various trapping scenarios, justifying robust performance of BINN.
\end{abstract}



\section{Introduction}
\label{sec:intro}

The macroscopic world is governed by the collective behavior of microscopic particles of two fundamental species: fermions and bosons, and this paper focuses on computations of the latter. Not obeying the Pauli exclusion principle, multiple bosons can occupy the same quantum state, leading to permutation-symmetry in the wave function. Mathematically, the (dimensionless) Hamiltonian for describing $N$ bosons trapped in a $d$-dimensional ($d=1,2,3$) external potential $v_{\rm ext}$ reads \cite{lieb2005bosons}
\begin{equation}\label{H original}
  \hat H_N = \sum_{i=1}^{N}\Bigl[-\tfrac12\nabla_i^{2} + v_{\rm ext}(\bm r_i)\Bigr]
           + \sum_{1\leq i<j\leq N} v_{\rm int}(\bm r_i-\bm r_j),
\end{equation}
where $\bm r_i=(r_{i,1},...,r_{i,d})\in\mathbb R^d$ is the position of the $i$-th boson 
and $v_{\rm int}$, commonly modeled \cite{bloch2008many} by a Dirac delta contact potential and practically regularized as a Gaussian 
mollifier \cite{pescia2022neural} 
\begin{equation}\label{eq:pair_pot}
  v_{\rm int}(\bm r)\equiv v_\sigma(\bm r):=gW_\sigma(\bm r),
  \qquad
  W_\sigma(\bm r):=\frac{\mathrm{e}^{-|\bm r|^2/\sigma^2}}
  {(\sigma\sqrt{\pi})^d},
\end{equation}
describes the short-range pairwise interaction between particles. 
Here, \(W_\sigma\rightharpoonup\delta^{(d)}\) as \(\sigma\to0\) with \(\sigma>0\) the interaction range, and \(g>0\) denotes the repulsive coupling strength.
The Hamiltonian $\hat H_N$ acts on the wave functions
$\Psi(\bm R)$ of the configuration
$\bm R=(\bm r_1,\ldots,\bm r_N)\in\mathbb R^{dN}$, where 
$|\Psi(\bm R)|^2$ is the probability density for observing 
$\bm R$, hence $\|\Psi\|_{L^2}=1$. At extremely low or zero temperature, the equilibrium properties of such a system are governed by its many-body ground state (GS), the lowest-energy quantum state giving 
\begin{equation}\label{eq:E0_intro}
  E_0 = \min_{\substack{\Psi\in H^1_{\rm sym}(\mathbb R^{dN})\\
                        \|\Psi\|_{L^2}=1}}
        \langle \Psi | \hat H_N | \Psi\rangle.
\end{equation}
Hereafter, we define
\[
  \begin{aligned}
    \langle \Psi | \hat H_N | \Psi\rangle
    ={}&
    \frac{1}{2}\sum_{i=1}^{N}
    \int_{\mathbb{R}^{dN}}
    \left|\nabla_i \Psi(\bm R)\right|^2\,\mathrm{d}\bm R
    \\
    &+
    \int_{\mathbb{R}^{dN}}
    \left[
      \sum_{i=1}^{N}v_{\mathrm{ext}}(\bm r_i)
      +
      \sum_{1\leq i<j\leq N}v_{\mathrm{int}}(\bm r_i-\bm r_j)
    \right]
    |\Psi(\bm R)|^2\,\mathrm{d}\bm R,
  \end{aligned}
\]  
and denote by
$H^1_{\rm sym}(\mathbb R^{dN})$ the permutation-symmetric
subspace of $H^1(\mathbb R^{dN})$, meaning that 
$\Psi(\bm r_{\pi(1)},\ldots,\bm r_{\pi(N)}) = \Psi(\bm r_1,\ldots,\bm r_N)$
for all $\pi\in S_N$ with 
$S_N$ the symmetric group of order $N!$. 
The high dimensionality of the configuration space for large \(N\)
and the stiffness caused by short-range interactions for small
\(\sigma\) make the direct computation of the wave function
challenging, which is the focus of this work.

Due to these obstacles, 
traditional methods on \eqref{eq:E0_intro} face a persistent tradeoff
between scalability and microscopic resolution. For instance, 
the exact diagonalization~\cite{lanczos1950iteration} and the tensor-network
methods~\cite{white1992density} are limited by the rapid growth of
the needed basis and the bond dimension, respectively, confining them
to a few particles or quasi-one-dimensional geometries. Although Diffusion Monte Carlo (DMC)~\cite{foulkes2001quantum,reynolds1982fixed} evades the grid explosion and provides near-exact benchmarks for bosons, its efficiency hinges on a well-crafted guiding function, which is typically obtained via a prior variational Monte Carlo (VMC) \cite{ceperley1977monte,mcmillan1965ground} calculation. 
Finite-temperature path-integral Monte Carlo scales to very large trapped
ensembles~\cite{krauth1996quantum}, but it targets thermal equilibrium rather than the zero-temperature wave function. 
Meanwhile, Gross–Pitaevskii (GP) theory~\cite{gross1961structure,lieb2005bosons,lieb2000derivation,pitaevskii1961vortex} reduces the $dN$-dimensional many-body problem to a nonlinear model for a single mean-field condensate orbital, which can be solved efficiently through diverse computational methods~\cite{altmann2021jmethod,antoine2017efficient,BaoCai,bao2025computing,danaila2010new,kong2026rotating}. It therefore serves as an essential asymptotic reference in the weakly interacting regime, but inherently lacks the explicit two-body correlations required to capture correlation energies and short-range pair structures.

As a successful application of deep neural networks for high-dimensional problems, neural quantum states (NQS)~\cite{carleo2017solving} that represent the full many-body wave function with neural networks under the traditional VMC framework have emerged as a promising paradigm for capturing GS, intensively developed in recent years for fermions  \cite{hermann2020deep,kim2024fermi,pfau2020ab}. 
For interacting bosons, the development of NQS 
remains more limited. 
It is challenging for a network to 
be effective as the interaction range decreases or the scale increases, especially in shrinking-range regimes \cite{lieb2005bosons} where the interaction length $\sigma$
is inversely proportional to $N$.
\cite{pescia2022neural} using DeepSets~\cite{zaheer2017deep} is restricted to periodic one- and
two-dimensional gases with at most $N=64$, and \cite{bedaque2024neural} resolves one-dimensional systems up to
$N=50$; 
Freitas and Vitiello~\cite{freitas2023bosenet} reach DMC accuracy 
for untrapped three-dimensional $^{4}$He clusters with $N\leq14$; and \cite{hul2026grand} accommodates variable particle numbers and treats two-dimensional trapped gases with $N\leq17$.
In summary, to our knowledge, existing continuous-space neural wave
functions have not yet demonstrated DMC-level accuracy for
trapped bosons in both two and three dimensions at the
$\mathcal{O}(10^2)$-particle scale. 
Also, it is desirable in applications for an
architecture to transfer robustly across confinement potentials.

To address these challenges in solving the GS problem \eqref{eq:E0_intro}, we construct the Boson-informed neural network (BINN). Our main contributions are as follows.
\begin{itemize}
\item The BINN is designed with a universal architecture based on a proposed Symmetry-Jastrow factorization for trapped interacting bosons, preserving permutation symmetry. The factorization gives rise to a novel network that separates into a one-body branch and a pair branch. The two branches integrate a cell Jastrow factor, which resolves the contact stiffness, and a potential-aware envelope, which captures the macroscopic confinement background, into a pooled structure that allows for closed-form local-energy evaluation.

\item Through numerical experiments, we validate BINN in terms of accuracy, scalability, and asymptotic consistency with theoretical results. The computed GS shows a trend $\mathcal{O}(N^{-1})$ toward mean-field limits in both fixed-range and shrinking-range regimes. The method scales to $N=350$ in the three-dimensional shrinking-range regime ($\sigma\approx0.03$) while maintaining DMC-level accuracy.

\item Necessity of the cell Jastrow factor and the potential-aware envelope is justified by ablation studies. Applications of BINN are made across structured and nonuniform trapping potentials, including ring-shaped, deep-lattice, double-well, and box traps.
\end{itemize}

The rest of the paper is organized as follows. Section~\ref{sec:background} provides the necessary preliminaries, including VMC. Section~\ref{sec:methods} presents the construction of BINN component
by component and describes its implementation. Numerical experiments and results are presented in Section~\ref{sec:baseline}.
\section{Preliminaries}
\label{sec:background}

This section formulates the problem, introduces the numerical framework underlying our method, and concludes with theoretical benchmarks that serve as diagnostics for the subsequent numerical results.

\subsection{Variational Monte Carlo (VMC) formulation}
\label{sec:vmc}

For the many-body problem of \eqref{H original}, we focus on the cases \(d=2\) and \(d=3\) dimensions (2D, 3D). The external potential
\(v_{\rm ext}\) is assumed confining, i.e.,
\(\lim_{|\bm r|\to\infty}v_{\rm ext}(\bm r)=\infty\), which ensures that the particles remain spatially localized, needed also for theoretical benchmarks later. 

The GS energy $E_0=E_0(N)$ defined by~\eqref{eq:E0_intro} involves the normalization
constraint $\|\Psi\|_{L^2}=1$.
Such a
constraint can be removed for convenience of computations by using equivalently the Rayleigh quotient:  
\begin{equation}\label{eq:rayleigh}
  E_0(N)
  = \min_{\substack{\Psi\in H^1_{\rm sym}(\mathbb R^{dN})\\ \Psi\neq 0}}
    E[\Psi],
  \qquad
  E[\Psi] := \frac{\langle\Psi|\hat H_{N}|\Psi\rangle}
                  {\|\Psi\|_{L^2}^2}.
\end{equation}
Under the setup, the minimum in~\eqref{eq:rayleigh} can be 
attained as theoretically analyzed~\cite{reed1978analysis}, and we denote the normalized minimizer by \(\Psi_0\), which satisfies the many-body
eigenvalue problem
\begin{equation}
  \hat H_{N}\Psi_0 = E_0(N)\Psi_0,\qquad \|\Psi_0\|_{L^2}=1.
  \label{eq:eigenproblem}
\end{equation}

Direct discretization of \eqref{eq:rayleigh}, a $dN$-dimensional
minimization, certainly becomes infeasible when the number of particles $N$ is large. Alternatively, from a Monte Carlo perspective, the energy functional in~\eqref{eq:rayleigh} can be expressed as an expectation~\cite{foulkes2001quantum,mcmillan1965ground}. For a sufficiently regular nonzero \(\Psi\), define the local energy \(E_L\) and the Born distribution \(\rho[\Psi]\) by
\begin{equation}\label{EL def}
  E_L[\Psi](\bm{R}) := \frac{\hat{H}_{N}\Psi(\bm{R})}{\Psi(\bm{R})},
  \quad
  \rho[\Psi](\bm{R}):=\frac{|\Psi(\bm{R})|^2}{\|\Psi\|_{L^2}^2},\quad 
  \bm R=(\bm r_1,\ldots,\bm r_N)\in\mathbb R^{dN},
\end{equation}
then 
\[
  E[\Psi]
  = \frac{\int_{\mathbb{R}^{dN}}\overline{\Psi(\bm{R})}\,
        \hat H_{N}\Psi(\bm{R})\,d\bm{R}}
         {\|\Psi\|_{L^2}^2}
  = \int_{\mathbb{R}^{dN}} \rho[\Psi](\bm{R})\,
        \frac{\hat H_{N}\Psi(\bm{R})}{\Psi(\bm{R})}\,d\bm{R}
  = \mathbb{E}_{\rho[\Psi]}\!\left[E_L[\Psi]\right].
\]
The energy of any state can thus be estimated by drawing configurations
$\bm R$ from $\rho[\Psi]$ and then averaging $E_L$, where the $dN$-dimensional quadrature
is replaced by sampling from a distribution determined by the state itself.
On this basis, VMC restricts the minimization to a
parameterized manifold of trial states $\Psi_\theta$, with parameters $\theta\in\mathbb{R}^P$ for some $P\in\mathbb{N}_+$, 
not necessarily normalized. The minimum of $E[\Psi_\theta]$ over $\theta$ could approach $E_0(N)$ from above, and we will abbreviate
$\rho_\theta:=\rho[\Psi_\theta]$.

It follows from~\eqref{eq:rayleigh} or~\eqref{eq:eigenproblem} that the GS \(\Psi_0\) is defined only up to a global phase and can be chosen strictly positive in the interior of the configuration domain~\cite{simon1982schrodinger}. This motivates the
log-space parameterization for the trial state of VMC \cite{kessler2021artificial}:
\begin{equation}
  \Psi_\theta(\bm{R}) = \mathrm{e}^{f_\theta(\bm{R})},
  \label{eq:logspace}
\end{equation}
where $f_\theta:\mathbb{R}^{dN}\to\mathbb{R}$ will be a neural-network function in our study, 
making $\Psi_\theta>0$ automatic. To compute the local energy, we first define the gradient operator for the $i$-th particle as $\nabla_i := \nabla_{\bm r_i} = \left( \frac{\partial}{\partial r_{i,1}}, \ldots, \frac{\partial}{\partial r_{i,d}} \right)$. The corresponding $dN$-dimensional Laplacian is $\nabla^2 := \sum_i \nabla_i^2$. By applying the differential identity $\nabla_i^2 \mathrm{e}^{f} = \mathrm{e}^{f}(\nabla_i^2 f + \vert{}\nabla_i f\vert{}^2)$, the local energy evaluates directly to
\begin{equation}
  E_L[\Psi_\theta](\bm{R})
  = -\frac{1}{2}\!\bigl(\nabla^2 f_\theta(\bm{R}) + |\nabla f_\theta(\bm{R})|^2\bigr)
    + V_{\rm ext}(\bm{R}) + V_{\rm int}(\bm{R}),
  \label{eq:eloc}
\end{equation}
where $V_{\rm ext}(\bm{R}):=\sum_i v_{\rm ext}(\bm{r}_i)$ and
$V_{\rm int}(\bm{R}):=\sum_{i<j}v_{\rm int}(\bm{r}_i-\bm{r}_j)$.

Substituting the exact GS $\Psi_0$ that satisfies \eqref{eq:eigenproblem} into the definition of $E_L$ in \eqref{EL def} makes the local energy configuration-independent, i.e., $E_L[\Psi_0](\bm{R})\equiv E_0,\,\forall \bm R$, and so $\operatorname{Var}[E_L[\Psi_0]]=0$. Conversely, if a $\Psi_\theta$ within the  manifold~\eqref{eq:logspace}  satisfies 
$\operatorname{Var}[E_L[\Psi_\theta]]=0$, then $\Psi_\theta$ must be the GS. Indeed, the vanishing variance indicates a constant $E_L[\Psi_\theta](\bm R)$ in $\bm R$ and $\hat{H}_N\Psi_\theta(\bm R)= E_L[\Psi_\theta]\Psi_\theta(\bm R)$. The positivity of $\Psi_\theta$ precludes its orthogonality to the positive GS, so $E_L[\Psi_\theta]$ equals $E_0$ making $\Psi_\theta$ the GS. 
 We will thus use the variance as a 
 complementary convergence diagnostic in Section~\ref{sec:ablation}.
In total, the VMC formulation for the GS problem reads:
\begin{equation}\label{VMC def}
  \min_{\substack{\theta\in\mathbb R^{P}\\ \Psi_\theta\in H^1_{\rm sym}\setminus\{0\}}}
    \mathbb{E}_{\rho_\theta}\!\left[E_L[\Psi_\theta]\right].
\end{equation}
It has been used in early works with $f_\theta$ taken as analytic Jastrow functions \cite{mcmillan1965ground}, while recent work has used neural-network trial states~\cite{pescia2022neural}.

\subsection{Optimization method}
\label{sec:sr} 
We will update the parameters $\theta$  for VMC by
using the Stochastic Reconfiguration
(SR)~\cite{amari1998natural,sorella1998green}, a standard
natural-gradient optimizer for neural quantum
states~\cite{hermann2020deep,pfau2020ab,stokes2020quantum,vicentini2022netket} briefly described below. 

For the energy gradient needed to update \(\theta=(\theta_k)_{1\times P}\in\mathbb{R}^P\), the fact 
\(\partial_{\theta_k}\Psi_\theta=(\partial_{\theta_k}f_\theta)\Psi_\theta\) from \eqref{eq:logspace}
and the self-adjointness of \(\hat H_N\) give the identity:  
\begin{align}
 \partial_{\theta_k}\mathbb E_{\rho_\theta}[E_L[\Psi_\theta]]
 &=\frac{2\langle\partial_{\theta_k}\Psi_\theta,
  (\hat H_N-\mathbb E_{\rho_\theta}[E_L[\Psi_\theta]])\Psi_\theta\rangle}
 {\langle\Psi_\theta,\Psi_\theta\rangle} \notag\\
 &=2\,\mathbb E_{\rho_\theta}\!\left[
   \bigl(\partial_{\theta_k}f_\theta-\mathbb E_{\rho_\theta}[\partial_{\theta_k}f_\theta]\bigr)
   \bigl(E_L[\Psi_\theta]
   -\mathbb E_{\rho_\theta}[E_L[\Psi_\theta]]\bigr)\right].
 \label{eq:energy_gradient}
\end{align}
At a fixed $\theta$, the expectations can be evaluated by Monte Carlo, where we adopt the Metropolis-adjusted Langevin algorithm
(MALA)~\cite{roberts1996exponential,roberts1998optimal} to generate from $\rho_\theta$ a batch of configurations 
$\{\bm R^{(s)}\}_{s=1}^B$ for $B\in\mathbb{N}_+$. 
Denote
\[
   O_k^{(s)}
  = \partial_{\theta_k} f_\theta(\bm{R}^{(s)})
    - \frac{1}{B}\sum_{m=1}^B\partial_{\theta_k}f_\theta(\bm{R}^{(m)}),
  \qquad
  e^{(s)} = E_L[\Psi_\theta](\bm{R}^{(s)}) - \bar{E}_L,
\]
where $\bar{E}_L=\frac{1}{B}\sum_{m=1}^B E_L[\Psi_\theta](\bm{R}^{(m)})$, 
and we collect these quantities into $\bm{O}_c=\big(O_k^{(s)}\big)_{B\times P}\in\mathbb{R}^{B\times P}$  and
$\bm{e}=\big(e^{(s)}\big)_{B\times 1}\in\mathbb{R}^B$. 
The energy gradient from \eqref{eq:energy_gradient} is estimated by
$
  \nabla_{\theta}\mathbb E_{\rho_\theta}
   [E_L[\Psi_\theta]]\approx
  \widehat{\nabla_\theta E[\Psi_\theta]}
  := \frac{2}{B}\,\bm{O}_c^\top\bm{e}.$ 
Let additionally $S := \frac{1}{B}\,\bm{O}_c^\top\bm{O}_c$ be the Fisher-information matrix~\cite{pescia2022neural} up to a constant factor, and then SR updates $\theta$ by
solving a $P\times P$ linear system:
\begin{equation}
  (S+\lambda I)\,\delta\theta
  =-\eta\,\widehat{\nabla_\theta E[\Psi_\theta]},
  \qquad \theta\leftarrow\theta+\delta\theta.
  \label{eq:sr_update}
\end{equation}
Here, $\eta>0$ is the learning rate and $\lambda>0$ is the diagonal regularization parameter. At \(\lambda=0\), SR projects imaginary-time evolution onto the variational
manifold and gives its natural-gradient direction~\cite{pescia2022neural}.

For better approximation in practice, $\Psi_\theta$ is often over-parametrized \cite{ding2025convergence}, namely $B\ll P$. 
Noting that 
$S$ has the highest rank $B-1$, the
push-through identity~\cite{woodbury1950inverting}
\[
  \Bigl(\tfrac1B\,\bm{O}_c^\top\bm{O}_c+\lambda I_P\Bigr)^{\!-1}\bm{O}_c^\top
  \;=\;
  \bm{O}_c^\top\Bigl(\tfrac1B\,\bm{O}_c\bm{O}_c^\top+\lambda I_B\Bigr)^{\!-1}
\]
turns \eqref{eq:sr_update} equivalently into
\[
  \delta\theta
  = -\eta\,(S+\lambda I)^{-1}\,\widehat{\nabla_\theta E[\Psi_\theta]}
  = -\eta\,\frac{2}{B}\,\bm{O}_c^\top\bigl(\bm{G}+\lambda I_B\bigr)^{-1}\bm{e},
  \qquad \bm{G}=\tfrac1B\bm{O}_c\bm{O}_c^\top,
\]
which requires solving a \(B\times B\) linear system instead of a \(P\times P\) one. Resampling with the updated wave function then completes one SR iteration. 

\subsection{Some theoretical benchmarks}
\label{sec:asymptotics}
We provide some theoretical results of \eqref{eq:E0_intro} from the literature, which will serve as quantitative benchmarks to validate our method in
Section~\ref{sec:baseline}. Two regimes will be considered here: a fixed-range mean-field
regime by Hartree theory and a shrinking-range regime by GP theory, where the coupling strength constant $g$ in \eqref{eq:pair_pot} from now on is adopted as the mean-field coupling scaling 
\[g=\frac{g_{\rm GP}}{N},\]
for a given macroscopic coupling parameter $g_{\rm GP}>0$. 

\textbf{Fixed-range regime.} 
The fixed-range regime refers to a fixed (can be arbitrarily small but independent of $N$) $\sigma>0$ in \eqref{eq:pair_pot}. By restricting the admissible state for the energy functional to the product
state $\Psi(\bm{R})=\prod_{i=1}^N\phi(\bm{r}_i)$, one can check that \eqref{eq:rayleigh} reduces to 
\begin{equation}\label{E_MF def}
E_{\rm MF}(N) = \min_{\|\phi\|=1}N\left[ \langle\phi\,|\!-\!\tfrac{1}{2}\nabla^2\!+\!v_{\rm ext}|\phi\rangle + \tfrac{N-1}{2}\iint|\phi(\bm{r})|^2 v_\sigma(\bm{r}-\bm{r}')|\phi(\bm{r}')|^2\,d\bm{r}\,d\bm{r}' \right].
\end{equation}
As
$N\to\infty$, the mean-field coupling $(N-1)g\to g_{\rm GP}$, so
$E_{\rm MF}(N)/N$ converges to 
\[
  e_{\rm Hartree}
  = \min_{\|\phi\|=1}\left[
      \langle\phi\,|\!-\!\tfrac{1}{2}\nabla^2\!+\!v_{\rm ext}|\phi\rangle
      + \tfrac{g_{\rm GP}}{2}\iint|\phi(\bm{r})|^2 W_\sigma(\bm{r}-\bm{r}')|\phi(\bm{r}')|^2\,d\bm{r}\,d\bm{r}'
    \right],
\]
formally, which is the Hartree minimum. This is rigorously stated as follows. 


\begin{theorem}[Hartree theory~\cite{lewin2014derivation}]
\label{thm:meanfield}
Consider a confining \(v_{\rm ext}\) and a Gaussian potential \eqref{eq:pair_pot} with $\sigma>0$ fixed in \eqref{eq:rayleigh}. Then, 
\begin{equation}\label{thm LNR eq}
  \lim_{N\to\infty}\frac{E_{\rm MF}(N)}{N}
  =e_{\rm Hartree}(\sigma),
  \qquad
  \lim_{N\to\infty}
  \frac{E_{\rm MF}(N)-E_0(N)}{N}=0.
\end{equation}
Moreover, denoting the exact single-particle density as
\begin{equation}\label{single-partical density}
  \rho_N(\bm{r})
  := \int_{\mathbb{R}^{d(N-1)}}
    |\Psi_0(\bm{r},\bm{r}_2,\dots,\bm{r}_N)|^2\,
    d\bm{r}_2\cdots d\bm{r}_N,
\end{equation}
and $\varphi_{\rm H}$ as the Hartree minimizer attaining $e_{\rm Hartree}$, we have the convergence
\[
  \rho_N \longrightarrow \rho_{\rm H}:=|\varphi_{\rm H}|^2\quad \mbox{in}\quad L^1\quad \mbox{as} \quad N\to\infty.
\]
\end{theorem}

Under the additional
assumptions of the expansion, Bogoliubov
theory~\cite{grech2013excitation,lewin2015bogoliubov} gives the sharper
finite-size correction 
\begin{equation}
  e_{\rm Hartree}(\sigma)-\frac{E_0(N)}{N}
  = \frac{C_H}{N} + \mathcal{O}(N^{-\frac{3}{2}}),
  \label{eq:Ecorr_expand}
\end{equation}
where the constant $C_H>0$ can be  computed from $d, v_{\rm ext}, g_{\rm GP}$ and $\sigma$. It refines the energy convergence implied by
\eqref{thm LNR eq}, and equivalently,
\begin{equation}
     N\left|e_{\rm Hartree}(\sigma)-\frac{E_0(N)}{N}
  -\frac{C_H}{N}\right|=\mathcal O(N^{-1/2}),
  \label{eq:Ecorr_trans}
\end{equation}
which forms a quantitative benchmark for our test. 

\textbf{Shrinking-range regime.}  
Next, we consider the shrinking-range regime in three dimensions. Taking the width 
$\sigma=c/N$ with a fixed $c>0$ in \eqref{eq:pair_pot}, we have
\begin{equation}\label{shrink v}
  v_{\sigma}(\bm{r})
  = N^2\,V(N\bm{r}),
  \quad\mbox{with}\quad
  V(\bm{s})=\frac{g_{\rm GP}}{(c\sqrt{\pi})^{3}}\,\mathrm{e}^{-|\bm{s}|^2/c^2}.
\end{equation}
The scattering length $a(v_\sigma)$, a constant dependent on $v_\sigma$ and defined as \cite{lieb2005bosons}: 
\[
  \psi(\bm{r})\longrightarrow 1-\frac{a(v_\sigma)}{|\bm{r}|}
  \quad(|\bm{r}|\to\infty),\qquad
  -\nabla^2\psi+v_\sigma\,\psi=0,
\]
can be verified to satisfy the scaling equation: 
$
  Na(v_{\sigma})=a(V)=:a_0.$ 
Such a scaling invariance leads to the following theorem, showing that in the limit $N\to \infty$, the GS energy per particle in 3D is given by the minimum of the GP functional: 
\[
  e_{\rm GP}
  = \min_{\|\phi\|=1}\left[
      \langle\phi\,|\!-\!\tfrac{1}{2}\nabla^2\!+\!v_{\rm ext}|\phi\rangle
      + \tfrac{g_{\rm eff}}{2}\int|\phi(\bm{r})|^4\,d\bm{r}
    \right],
\]
with an effective coupling constant
$g_{\rm eff}:=4\pi a_0$~\cite{lieb2005bosons}.  

\begin{theorem}[GP theory~\cite{lieb2005bosons}]
\label{thm:lieb_seiringer}
Set $d=3$ in \eqref{eq:rayleigh}, and consider a confining \(v_{\rm ext}\) and the shrinking-range potential 
\eqref{shrink v} with a fixed $c>0$. Then, 
\begin{equation}\label{GP limit}
  \lim_{N\to\infty}\frac{E_0(N)}{N} = e_{\rm GP}.
\end{equation}
Moreover, with $\varphi_{\rm GP}$ the GP minimizer attaining $e_{\rm GP}$, 
the single-particle density \eqref{single-partical density} converges as 
\[
  \rho_N \longrightarrow |\varphi_{\rm GP}|^2\quad \mbox{in}\quad 
  L^1(\mathbb{R}^3)\quad \mbox{as} \quad N\to\infty.
\]
\end{theorem}

Motivated by trapped Bogoliubov
theories~\cite{brennecke2022bogoliubov} for compactly supported
interactions, we examine the empirical form, 
\begin{equation}\label{eq:bogoliubov_scaling}
  \frac{E_0(N)}{N} = e_{\rm GP}+  \mathcal{O}(N^{-1})+\mathcal{O}(N^{-2}).
\end{equation}

\section{Boson-informed neural network (BINN) method}
\label{sec:methods}

We now construct our BINN to solve GS under VMC~\eqref{VMC def}.  We  first give its overall structure and then detail the key components. The full algorithm and efficiency are discussed in the end.


\subsection{Basic network architecture}
\label{sec:ansatz} 
Motivated by the Slater--Jastrow factorization~\cite{ceperley1977monte,mcmillan1965ground} for describing fermions and the statistical symmetry of bosons,  
we adopt a \textit{Symmetry-Jastrow} factorization as the foundation for 
the BINN trial state $\Psi_\theta\approx\Psi$: 
\vspace{-4mm}
\begin{equation}
  \Psi_\theta(\bm{R}) = \Phi^{\rm sym}_{\theta_1}(\bm{R})\,\Phi^{\rm Jas}_{\theta_2}(\bm{R}),
  \quad
  \Phi^{\rm Jas}_{\theta_2}(\bm{R}) = \exp\!\left(\sum_{i<j}u_{\theta_2}(r_{ij})\right),\quad r_{ij}=|\bm{r}_i-\bm{r}_j|.
  \label{eq:psi_factored}
\end{equation}
Here, the complete set of trainable parameters is defined as $\theta := \{\theta_1, \theta_2\}$. 
The symmetry component $\Phi^{\rm sym}_{\theta_1}$ is designed to capture macroscopic features of bosons, and the Jastrow component 
$\Phi^{\rm Jas}_{\theta_2}$ is for capturing correlation effects between bosons. 
The pair function $u_{\theta_2}$ of 
the Jastrow exponent will be  constructed in Section~\ref{sec:cusp_jastrow}.  The symmetric wave function $\Phi^{\rm sym}_{\theta_1}$ is a log-space construction to fulfill the positivity of GS, composed of a potential-aware one-body envelope sum and a pooled neural backbone: 
\vspace{-4mm}
\begin{equation}
\log\Phi^{\rm sym}_{\theta_1}(\bm{R})
= \sum_{i=1}^N \chi^{\rm env}_{\theta_1}(\bm{r}_i) + f^{\rm TNN}_{\theta_1}(\bm{R}),
\label{eq:sym_split}
\end{equation}
where the one-body envelope $\chi^{\rm env}_{\theta_1}$ incorporates the macroscopic confinement background into the representation of each boson and the backbone $f^{\rm TNN}_{\theta_1}$ is responsible for the deviation of particles from the envelope.
We construct $f^{\rm TNN}_{\theta_1}$ in the following with $\chi^{\rm env}_{\theta_1}$
left to Section~\ref{sec:envelope}. Overall, the BINN architecture \eqref{eq:psi_factored} is consistent with the VMC trial \eqref{eq:logspace}, where now 
$f_\theta(\bm{R})=\sum_{i=1}^N \chi^{\rm env}_{\theta_1}(\bm{r}_i) + f^{\rm TNN}_{\theta_1}(\bm{R})+\sum_{i<j}u_{\theta_2}(r_{ij}).$

We construct the neural backbone from a positive, permutation-symmetric tensor neural network (TNN), which will thus make the full BINN \eqref{eq:psi_factored} permutation-symmetric. 
A standard TNN representation~\cite{wang2024tensor} approximates a high-dimensional function as a sum of separable products, 
and here on the many-body function, we group the coordinates by particles so that
each factor takes the full $d$-dimensional position of one particle:
\vspace{-1mm}
\[
  \Phi_{\rm TNN}(\bm{R})=\Phi_{\rm TNN}(\bm{r}_1,\ldots,\bm{r}_N)
  \;=\; \sum_{m=1}^R \prod_{i=1}^N \phi_{m,i}(\bm{r}_i),
\]
with an ordered tuple of $N$ general neural
factors $\phi_{m,i}:\mathbb{R}^d\to\mathbb{R}$ for each branch $m$. For a bosonic GS here, two special adaptations are made. 
First, to obtain the permutation
symmetry, we impose a shared factor across all particle slots, i.e.,  
$\phi_{m,i}\equiv\phi_m,\ 1\leq i\leq N$. The exchange symmetry is therefore exact and inherent, without large sums of $N$-dependent complexity. 
Second, to ensure the positivity, we impose $\phi_{m}>0$ by writing
$\phi_m=\mathrm{e}^{\varphi_m}$ with neural log-orbitals
$\varphi_m:\mathbb{R}^d\to\mathbb{R}$. With an explicit positive branch weight
$\mathrm{e}^{c_m}$, where $c_m\in\mathbb{R}$ is trainable, taking the logarithm of the resulting TNN gives
\begin{equation}
   f^{\rm TNN}_{\theta_1}(\bm{R})
  \;=\; \log\sum_{m=1}^R \exp\!\Bigl(c_m+\sum_{i=1}^N \varphi_m(\bm{r}_i)\Bigr).
  \label{eq:tnn}
\end{equation}

Each branch in \eqref{eq:tnn} is a Hartree product, while $R>1$ mixes such
products to capture non-product effects. Defining the log-sum-exp function $\operatorname{LSE}(\{z_m\}_{m=1}^R):=\log\sum_{m=1}^{R}\exp(z_m)$, the shared envelope
can equivalently be placed inside the pooled branches:
\begin{equation}
\sum_{i=1}^N\chi^{\rm env}_{\theta_1}(\bm r_i)+f^{\rm TNN}_{\theta_1}(\bm R)
=\operatorname{LSE}\left(\bigg\{
c_m+\sum_{i=1}^N\left[\varphi_m(\bm r_i)+\chi^{\rm env}_{\theta_1}(\bm r_i)\right]
\bigg\}_{m=1}^R\right).
\label{eq:tnn_env_lse}
\end{equation}
This pooled form would enable the exact kinetic-trace evaluation given in
Section~\ref{sec:kinetic}, and the complete architecture is demonstrated in
Fig.~\ref{fig:architecture}.

\begin{figure}[t!]
\centering
\resizebox{\textwidth}{!}{%
\begin{tikzpicture}[
  >={Stealth[length=2mm]},
  box/.style={draw, rounded corners, minimum height=8.5mm, minimum width=38mm,
              align=center, font=\small, inner sep=3pt},
  neural/.style={box, fill=blue!8},
  analytic/.style={box, fill=orange!16},
  op/.style={draw, circle, minimum size=7mm, inner sep=0pt,
             font=\footnotesize, fill=black!6},
  io/.style={align=center, font=\small},
  scalehead/.style={font=\normalsize\sffamily\bfseries, text=black!90,
                   text width=30mm, align=center, inner sep=0pt},
  grp/.style={draw, dashed, rounded corners, fill=black!3, inner sep=5pt}
]
\node[io] (coord) at (0,0) {$\bm R=$\\$(\bm r_1,\dots,\bm r_N)$};
\node[analytic] (env) at (4.2,1.0)
  {envelope basis\\$\{v_\gamma\}_{\gamma=1}^{L}$};
\node[neural] (mlp) at (4.2,2.3)
  {TNN log-orbitals\\$\varphi_m(\bm r_i),\quad m=1,\ldots,R$};
\node[op] (m1) at (7.0,1.65) {$+$};
\node[op,minimum size=15mm] (s1) at (8.6,1.65) {$c_m+\sum_i$};
\node[box,fill=black!6,minimum width=20mm] (lse) at (10.8,1.65)
  {$\operatorname{LSE}(\,\cdot\,)$};
\node[scalehead] (macroscale) at (4.2,3.15)
  {one-body branch};
\node[analytic] (uc) at (4.2,-2.3) {cell Jastrow\\$u_{\rm c}(r_{ij})$};
\node[neural] (unn) at (4.2,-1.0)
  {neural Jastrow\\$\varphi_{\theta_2}(r_{ij}^2)$};
\node[op] (m2) at (7.0,-1.65) {$+$};
\node[op,minimum size=15mm] (s2) at (8.6,-1.65) {$\sum_{i<j}$};
\node[scalehead] (pairscale) at (4.2,-3.15)
  {pair branch};
\node[op] (plus) at (12.8,0) {$+$};
\node[io,font=\normalsize] (out) at (14.3,0)
  {$\log\Psi_\theta$};
\begin{scope}[on background layer]
  \node[grp, fit=(env)(mlp)(macroscale)] (g1) {};
  \node[grp, fit=(uc)(unn)(pairscale)]  (g2) {};
\end{scope}
\draw[->] (coord) to[out=70,in=180]
  node[above,font=\footnotesize,pos=.55]{$\bm r_i$} (g1.west |- m1);
\draw[->] (coord) to[out=-70,in=180]
  node[below,font=\footnotesize,pos=.55]{$r_{ij}$} (g2.west |- m2);
\draw[->] (env) -| (m1);
\draw[->] (mlp) -| (m1);
\draw[->] (m1) -- (s1);
\draw[->] (s1) -- (lse);
\draw[->] (lse) to[out=0,in=110] (plus);
\draw[->] (uc)  -| (m2);
\draw[->] (unn) -| (m2);
\draw[->] (m2) -- (s2);
\draw[->] (s2) to[out=0,in=-110] (plus);
\draw[->] (plus) -- (out);
\end{tikzpicture}}
\caption{{The BINN architecture. At each scale an analytic prior and a
neural network act in parallel and are summed, the one-body branch
pooling over particles into a log-sum-exp readout and the pair
branch summing over pair distances. Orange blocks mark the analytic
constructions, blue blocks the neural networks.}}
\label{fig:architecture}
\end{figure}
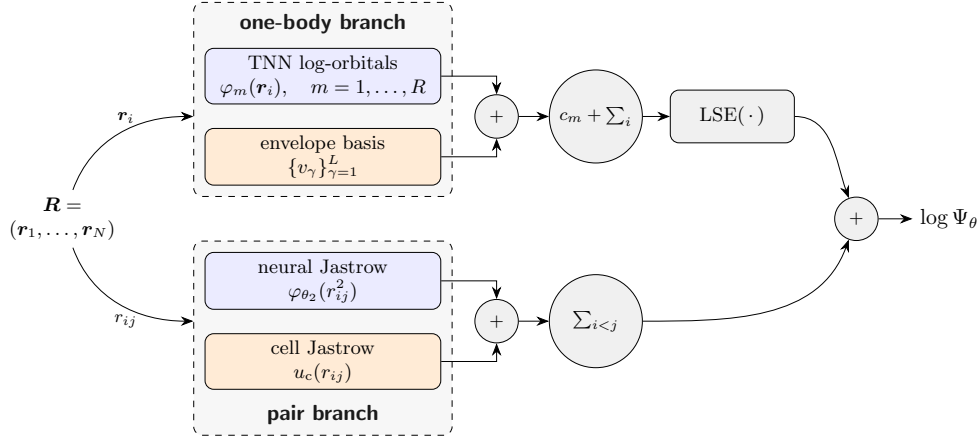

\subsection{Potential-aware one-body envelope}
\label{sec:envelope} 
In principle, the TNN backbone in \eqref{eq:sym_split} alone has the capacity to approximate a wave profile. However, complex external confinement potentials, i.e., $v_{\mathrm{ext}}$ in \eqref{H original}, can induce highly structured macroscopic spatial profiles in the target wave function. For instance, features like the bimodal structure of a double well, the periodic modulation of an optical lattice, and the annular geometry of a ring trap are difficult for a generic network to learn from scratch. 
To improve accuracy and training efficiency, we introduce a shared, potential-aware envelope \(\chi^{\mathrm{env}}_{\theta_1}\) 
  via \eqref{eq:sym_split}. This envelope analytically encodes the confinement structure, allowing the neural backbone to focus solely on refining the remaining many-body behavior. 

The construction of $\chi^{\rm env}_{\theta_1}$ is heuristic. The first observation is that for the one-body Hamiltonian
$-\tfrac12\nabla^2+v_{\rm ext}$
with the most commonly used harmonic oscillator trapping potential 
\begin{equation}\label{harmonic oscillator trapping potential}
    v_{\rm ext}=\tfrac12\omega^2\vert{}\bm{r}\vert{}^2,
\end{equation}
where $\omega>0$ denotes the frequency, 
the single-particle GS $\varphi_\star$ is a Gaussian. Therefore, its logarithm up to an additive constant satisfies
\vspace{-2mm}
\[\log\varphi_\star = -\tfrac12\omega\vert{}\bm{r}\vert{}^2 = -\frac{v_{\rm ext}(\bm{r})}{\omega}.\] For a general trap $v_{\rm ext}$, the closed-form expression for $\varphi_\star$ may not be available, while physical intuition suggests that a particle in a stationary state tends to localize in regions of lower potential energy. This largely implies an inverse relationship between $v_{\rm ext}$ and the amplitude of the wave function, and we encode this into the wave function through a factor $\mathrm{e}^{-\alpha v_{\rm ext}}$ with $\alpha>0$. In practice, 
multiple external trapping effects can be considered, and the external potential in this case is given in an additive form
\vspace{-1mm}
\[v_{\rm ext}(\bm r)=\sum_{\gamma=1}^{L}v_{\gamma}(\bm r).\]
To be general, 
we thus define the envelope $\chi^{\rm env}_{\theta_1}$ for the log-setup \eqref{eq:sym_split} as
\vspace{-2mm}
\begin{equation}
\chi^{\rm env}_{\theta_1}(\bm{r}) =  -\sum_{\gamma=1}^{L}\alpha_{\gamma}v_{\gamma}(\bm{r}), \qquad \alpha_{\gamma}\geq 0, \quad \gamma=1,\ldots,L.
\label{eq:envelope}\end{equation}
Here, all coefficients $\{\alpha_{\gamma}\}_{\gamma=1}^{L}$ are treated as trainable parameters, allowing the envelope to adapt independently to different contributions to the confinement. Sections~\ref{sec:ablation} and \ref{sec:geometry_stress} will demonstrate the necessity and effectiveness of this envelope design.

This completes the one-body branch introduced in \eqref{eq:psi_factored}.
In total, the parameter set $\theta_1$ comprises three parts: the envelope coefficients $\{\alpha_{\gamma}\}_{\gamma=1}^{L}$, the weights $\{c_m\}_{m=1}^{R}$ in \eqref{eq:tnn}, and all internal parameters of $\varphi_m$.
\subsection{Cell Jastrow factor}
\label{sec:cusp_jastrow}
Within \eqref{eq:psi_factored}, 
the Jastrow component 
$\Phi^{\rm Jas}_{\theta_2}$ aims to capture the correlation between bosons through the key 1D function $u:[0,\infty)\to\mathbb{R}$, where we henceforth write \(u= u_{\theta_2}\) for brevity. We shall construct one that is effective for learning.  

As shown in \eqref{eq:pair_pot}, the interaction potential scales as $v_{\sigma}=\mathcal{O}(\sigma^{-d})$. For the exact GS, the local energy $E_L$ \eqref{eq:eloc} must remain constant in all spatial configurations $\bm{R}$. This constancy is particularly difficult to learn when two particles approach each other. To compensate for the diverging pairwise interaction, the trial wave function $\Psi_{\theta}$ requires a specialized local structure. Specifically, the kinetic energy of this local structure, given by $-\frac{1}{2}\!\bigl(\nabla^2 f_\theta(\bm{R}) + \vert{}\nabla f_\theta(\bm{R})\vert{}^2\bigr)$ in \eqref{eq:eloc}, should cancel the singularity in $v_{\mathrm{int}}$. Due to spectral bias~\cite{rahaman2019spectral}, neural networks struggle to efficiently capture such high-frequency features, so relying solely on a neural network to learn this cancellation is impractical.  We therefore propose to analytically incorporate this physical cancellation into the framework. To achieve this, we first construct an analytic cell factor $u_{\rm c}$ to absorb the singularity in $v_{\sigma}$. A smooth neural network is then employed as a residual to learn the regular part of the correlations.

To begin with, we expand the local energy~\eqref{eq:eloc} under the factorization~\eqref{eq:psi_factored}:
\[
  E_L = E_L^0 + R_{\rm 2b} + R_{\rm cross} + R_{\rm 3b},
\]
where 
$
  E_L^0 := -\tfrac12\bigl(\nabla^2\log\Phi^{\rm sym}_{\theta_1}+|\nabla\log\Phi^{\rm sym}_{\theta_1}|^2\bigr)+V_{\rm ext}(\bm{R})$ is independent of $u$, 
and 
\begin{align}
  &R_{\rm 2b} := \sum_{i<j}\bigl[-\nabla_{\rm rad}^2u(r_{ij})-|u'(r_{ij})|^2+v_\sigma(r_{ij})\bigr]\mbox{ with }\nabla_{\rm rad}^2u(r):=u''(r)+\frac{d-1}{r}u'(r),
  \notag\\
  &R_{\rm cross} := -\sum_{i<j}u'(r_{ij})(\nabla_i-\nabla_j)\log\Phi^{\mathrm{sym}}_{\theta_1}\cdot (\bm{r}_i-\bm{r}_j)/r_{ij},
  \notag\\
  &R_{\rm 3b} := -\tfrac12\sum_i\sum_{j\ne i}\sum_{k\ne i,j}
    \nabla_i u(r_{ij})\cdot\nabla_i u(r_{ik}).
  \notag
\end{align}
The interaction $v_\sigma$ enters the formulation exclusively through the term $R_{\rm 2b}$. To track the complete singular contribution of a fixed pair, write
$
  \bm F_{ij}:=u'(r_{ij})(\bm r_i-\bm r_j)/r_{ij}.$
The two-body term of the pair $(i,j)$ and all terms in
$R_{\rm cross}+R_{\rm 3b}$ that contain the potentially large factor
$\bm F_{ij}$ can then be collected exactly as
\vspace{-2mm}
\begin{equation}
  \mathcal C_{ij}
  :=\mathcal Q_\sigma[u](r_{ij})
  -\bm F_{ij}\cdot\left[
    (\nabla_i-\nabla_j)\log\Phi^{\rm sym}_{\theta_1}
    +\sum_{k\ne i,j}(\bm F_{ik}-\bm F_{jk})
  \right],
\label{eq:isolated_contact_contribution}
\end{equation}
where \begin{equation}
  \mathcal{Q}_\sigma[u](r) \;:=\;
  -\nabla_{\rm rad}^2\,u(r) - \bigl|u'(r)\bigr|^2 + v_\sigma(r),
  \qquad r\geq 0.
  \label{eq:cusp_condition}
\end{equation}

Let us examine the contact limit \(\bm r_i=\bm r_j\). Exchange symmetry implies \((\nabla_i-\nabla_j)\log\Phi^{\rm sym}_{\theta_1}=0\), while \(\bm F_{ik}=\bm F_{jk}\) for every \(k\ne i,j\). The two coupling terms in~\eqref{eq:isolated_contact_contribution} therefore vanish, reducing the contact contribution to \(\mathcal C_{ij}=\mathcal Q_\sigma[u](0)\). To control this contribution, we analyze \(\mathcal Q_\sigma[u](r)\) in a neighborhood of \(r=0\) and focus on a near-contact cell \([0,\ell]\), where \(\ell>0\) is fixed and independent of \(\sigma\). On this cell, we seek an analytic radial factor \(u_{\rm c}(r)=\log f(r)\), for which
\[
\mathcal Q_\sigma[u_{\rm c}](r)
=-\nabla_{\rm rad}^2u_{\rm c}(r)-|u_{\rm c}'(r)|^2+v_\sigma(r)
=\frac{-\nabla_{\rm rad}^2f(r)+v_\sigma(r)f(r)}{f(r)}
=:q_\sigma(f)(r).
\]We impose the following two structural requirements on this construction:
\begin{itemize}
\item \emph{Positivity}: $f$ is strictly positive;
\item \emph{Locality}: $f$ is constant beyond the cell, leaving
longer-range correlations to the smooth neural network.
\end{itemize}
To achieve this, we seek a solution corresponding to a constant right-hand side, i.e., $q_\sigma(f)(r)\equiv \nu_\sigma$ for $0\le r\le \ell$, which yields
$-\bigl(f''+\frac{d-1}{r}f'\bigr)+v_\sigma f=\nu_\sigma f, \quad r\in(0,\ell),$ where $\nu_\sigma$ is a constant. 
To ensure the regularity of $\nabla_{\rm rad}^2f(r)$ at the origin, we require $f'(0)=0$. Furthermore, the locality condition that requires $f$ to match smoothly a constant continuation at $r=\ell$ implies $f'(\ell)=0$. Combining these conditions yields a Neumann eigenvalue problem.
\begin{equation}
  \left\{
  \begin{array}{l}
    -\Bigl(f''+\tfrac{d-1}{r}\,f'\Bigr)+v_\sigma f
      \;=\; \nu_\sigma\,f, \qquad r\in(0,\ell),\\[4pt]
    f'(0)\;=\;f'(\ell)\;=\;0 .
  \end{array}
  \right.
  \label{eq:cell_problem}
\end{equation}
The radial operator \(-\nabla_{\rm rad}^2+v_\sigma\) in the cell with the boundary condition thus defines a
discrete spectrum by the classical Sturm--Liouville theory
\cite{courant1953methods}, and the positivity of $f$ selects the principal eigenpair
\((\nu_\sigma,f_\sigma)\). We normalize the principal eigenfunction by
\(f_\sigma(\ell)=1\). By construction, the cell factor $u_{\rm c}(r)=\log f_\sigma$ satisfies 
\begin{equation}
  \mathcal{Q}_\sigma[u_{\rm c}](r)\equiv\nu_\sigma,\qquad 0\le r\le\ell.
  \label{eq:cell_defect}
\end{equation}

Next, we combine the analytic cell factor $u_{\rm c}$ with a
smooth network remainder,
\begin{equation}
  u_{\theta_2}(r) \;=\; u_{\rm c}(r)
    + \underbrace{\varphi_{\theta_2}(r^2)}_{u_{\rm NN}(r)} ,
  \label{eq:jas_expression}
\end{equation}
where $\varphi_{\theta_2}$ is a one-dimensional MLP function with parameter $\theta_2$. With tanh activations and a compact parameter set $\Theta$, we have the uniform
derivative bound
$
M_\Theta:=
\sup_{\theta_2\in\Theta}\ \sup_{0\le s\le\ell^2}
\left(
|\partial_s\varphi_{\theta_2}(s)|
+|\partial_s^2\varphi_{\theta_2}(s)|
\right)<\infty.$
The squared variable input for $\varphi_{\theta_2}(s)$ in \eqref{eq:jas_expression}, i.e., $s=r^2$, provides smallness of its derivative near $r=0$, which is essential (will be justified via the ablation study later) to obtain the following uniform boundedness for $\mathcal{Q}_\sigma[u]$.

\begin{theorem}[Uniform contact-layer bound]
\label{prop:cusp_laplacian}
Let $\Theta$ be compact with $M_\Theta<\infty$, and fix $d(=2,3)$, $\ell>0$,
and $\sigma_0,\kappa>0$. For any $g>0$ and $0<\sigma\le\sigma_0$ satisfying
$g\sigma^{2-d}\le\kappa$, let $u_{\theta_2}$ be given by~\eqref{eq:jas_expression},
where $u_{\rm c}=\log f_\sigma$ and $f_\sigma$ is the principal eigenfunction
of~\eqref{eq:cell_problem} normalized by $f_\sigma(\ell)=1$.
Then we have
\[
\sup_{\theta_2\in\Theta}\ \sup_{0\le r\le\ell}
\bigl|\mathcal{Q}_\sigma[u_{\theta_2}](r)\bigr|\le C ,
\]
where $C=C(M_\Theta,\ell,\kappa,\sigma_0,d)>0$, independent of $g$, $\sigma$, and $\theta_2$.
\end{theorem}
\begin{proof}
See Appendix~\ref{app:proof_cusp_laplacian}.
\end{proof}

The condition $g\sigma^{2-d}\le\kappa$ is indeed reasonable in the shrinking-range regime by noting that $g=g_{\rm{GP}}/N$, $2-d=-1$ or 0, and $\sigma=c/N$. 
The role of $u_{\rm c}$ can be viewed as follows. Taking 
$u_{\rm NN}(r)=\varphi_{\theta_2}(r^2)$, we have 
\[
 \mathcal Q_\sigma[u_{\rm NN}](0)
 =\frac{g}{\pi^{d/2}\sigma^d}
  -2d\,\partial_s\varphi_{\theta_2}(0)=\frac{g\sigma^{2-d}}{\pi^{d/2}}\,\sigma^{-2}-2d\,\partial_s\varphi_{\theta_2}(0).
\]
This value can scale as $\mathcal O(\sigma^{-2})$ as $\sigma\to 0$, making the learning increasingly challenging. 
The analytic factor $u_{\rm c}$ provides a short-range prior before optimization, and
Theorem~\ref{prop:cusp_laplacian} shows that the remaining
contribution is uniformly bounded in $\sigma$ and $r$.  As a result, the neural network only needs to minimize a bounded residual rather than the full two-body correlation that is increasingly difficult to fit as $\sigma \to 0$. The resulting
training benefit is examined in the ablation study of
Section~\ref{sec:ablation}.
\subsection{Local-energy evaluation and training}
\label{sec:kinetic}
\label{sec:cost}
Based on the introduced architecture for BINN, 
we deduce here an evaluation of the local-energy. According to~\eqref{eq:eloc}, the kinetic term
for a general $f_\theta:\mathbb R^{dN}\to\mathbb R$ requires
$
  \nabla_{\bm R}^2 f_\theta
  =\sum_{i=1}^N\nabla_i^2 f_\theta.$ 
For BINN, \eqref{eq:psi_factored} and~\eqref{eq:sym_split} produce the following exact decomposition:
\begin{equation}
  \nabla^2 f_\theta(\bm R)
  =\underbrace{\sum_{i=1}^N
      \nabla_i^2\chi^{\rm env}_{\theta_1}(\bm r_i)}_{\text{one-body}}
   +\underbrace{2\sum_{i<j}
      \nabla_{\rm rad}^2u_{\theta_2}(r_{ij})}_{\text{two-body}}
   +\underbrace{\nabla^2 f^{\rm TNN}_{\theta_1}(\bm R)}_{\text{neural}}.
  \label{eq:lap_decomposition}
\end{equation}
The one-body and two-body terms can be evaluated directly through
$d$-dimensional and radial derivatives, respectively. The remaining
neural term, however, is the Laplacian of
$f_{\theta_1}^{\rm TNN}:\mathbb R^{dN}\to\mathbb R$.
A common automatic-differentiation implementation obtains this
Laplacian by first computing the $dN\times dN$ coordinate Hessian
$\operatorname{Hess}_{\bm R}f_{\theta_1}^{\rm TNN}$ and then taking
its trace~\cite{li2024forward}.

To accelerate this calculation, we
exploit the pooled structure of $f_{\theta_1}^{\rm TNN}$ to evaluate
its Laplacian exactly without forming the full coordinate Hessian. Denoting 
\[
  f^{\rm TNN}_{\theta_1}=\log\sum_{m=1}^R \mathrm{e}^{S_m+c_m},
  \qquad
  S_m=\sum_{i=1}^N\varphi_m(\bm{r}_i),
  \qquad
  p_m=\frac{\mathrm{e}^{S_m+c_m}}
  {\sum_{{l^\prime}=1}^R\mathrm{e}^{S_{l^\prime}+c_{l^\prime}}},
\]
the chain rule gives
$
  \nabla_{i}f^{\rm TNN}_{\theta_1}
  =\sum_{m=1}^R p_m\,\nabla_{i}\varphi_m(\bm{r}_i).$ 
It is direct to check that the derivative of $p_m$ with respect to
$\bm r_i$ can be written using only the $R$ local orbital gradients:
\[
\nabla_{i}p_m=p_m\Bigl(\nabla_{i}\varphi_m(\bm{r}_i)
  -\sum_{{l^\prime}=1}^R p_{{l^\prime}}\,\nabla_{i}\varphi_{{l^\prime}}(\bm{r}_i)\Bigr).
\]
Consequently, taking one more derivative and summing over
$i=1,\dots,N$ gives
\begin{equation}
  \nabla^2 f^{\rm TNN}_{\theta_1}
  =\sum_{i=1}^N\Bigl[\sum_{m=1}^R p_m\bigl(\nabla_i^2\varphi_m(\bm{r}_i)
      +|\nabla_i\varphi_m(\bm{r}_i)|^2\bigr)
    -\Bigl|\sum_{m=1}^R p_m\,\nabla_i\varphi_m(\bm{r}_i)\Bigr|^2\Bigr].
  \label{eq:exact_lap}
\end{equation}
Substituting~\eqref{eq:exact_lap} into~\eqref{eq:lap_decomposition}, and
then into~\eqref{eq:eloc}, gives the local energy $E_L[\Psi_\theta]$ needed for training. 
Equation~\eqref{eq:exact_lap} retains the full configuration dependence through the weights $p_m$, but replaces the calculation of the Laplacian of a $dN$-dimensional function with $R$ different $d$-dimensional functions evaluated at $N$ points. The other two terms in~\eqref{eq:lap_decomposition} are evaluated directly:
the one-body term requires $N$ $d$-dimensional Laplacians, while the
pair term uses the radial derivatives defined in
~\eqref{eq:cusp_condition} for all $N(N-1)/2$ pairs in one vectorized
batch. The pair calculation therefore retains
$\mathcal O(N^2)$ total work and determines the overall scaling in $N$.

\begin{algorithm}[h!]
\caption{BINN method.}
\label{alg:iteration}
\begin{algorithmic}[1]
\Require Physical parameters: dimension $d$, particles $N$, external potential $v_{\rm ext}$,
  macroscopic coupling $g_{\rm GP}$, range $\sigma$; 
  Computational parameters: cell radius $\ell$, number of configurations $B$, 
  total iterations $T$
\Ensure trained network parameters $\theta=\theta^{(T)}$; output the energy \eqref{numerical energy} and the GS wave function \eqref{eq:psi_factored} with $L^2$-normalization
\State Initialization: set $g\gets g_{\rm GP}/N$; precompute $u_{\rm c}=\log f_\sigma$ on $[0,\ell]$
  from~\eqref{eq:cell_problem}; 
initialize $\theta=\theta^{(0)}$; draw initial configurations
  $\{\bm{R}^{(s)}\}_{s=1}^B$ from a Gaussian distribution 
\For{$t=0,\dots,T-1$}
  \State update $\{\bm{R}^{(s)}\}_{s=1}^B$ targeting $|\Psi_\theta|^2$ by MALA  
  \State evaluate $E_L(\bm{R}^{(s)})$ using
    \eqref{eq:eloc} and~\eqref{eq:exact_lap}
  \State update by one SR step~\eqref{eq:sr_update} for 
  $\theta^{(t+1)}$ 
\EndFor
\end{algorithmic}
\end{algorithm}
\begin{table}[h!]
\small
\centering
\caption{Peak GPU memory and wall-clock time per full training step for
  two exact implementations of the TNN Laplacian on the 3D
  shrinking-range benchmark ($B=256$, $R=32$, $16$~GB device).
  `Direct' evaluates the $dN$ diagonal second derivatives together in
  one vectorized batch without using the pooling identity; `Pooled'
  uses~\eqref{eq:exact_lap}. OOM denotes out of memory.}
\label{tab:mem_trace}
\begin{tabular}{lrrrrrrr}
\toprule
$N$ & $50$ & $60$ & $70$ & $80$ & $90$ & $100$ & $150$ \\
\midrule
Direct memory (GB) & $2.03$ & $2.92$ & $3.96$ & $5.16$ & $6.52$ & $8.04$ & OOM \\
Pooled memory (GB) & $0.37$ & $0.50$ & $0.65$ & $0.82$ & $1.01$ & $1.23$ & $2.60$ \\
\midrule
Direct time (ms) & $45$ & $58$ & $76$ & $96$ & $120$ & $146$ & / \\
Pooled time (ms) & $18$ & $21$ & $25$ & $30$ & $35$ & $42$ & $88$ \\
\bottomrule
\end{tabular}
\end{table}


Algorithm~\ref{alg:iteration} summarizes the complete procedure of the presented BINN method, including the MALA
sampling for configurations and the SR update for optimization described in
Section~\ref{sec:sr}. The detailed setup of MALA and SR in our implementation for experiments is reported in Appendix~\ref{app:computational}. 
The algorithm outputs the trained $\theta$ for the normalized network wave function $\Psi_\theta$ \eqref{eq:psi_factored} and the numerical energy $E[\Psi_\theta]$ which under VMC \eqref{VMC def} can be computed here as 
 \begin{equation}\label{numerical energy}
 E[\Psi_{\theta}]\approx\frac{1}{B}\sum_{s=1}^B E_L[\Psi_{\theta}](\bm{R}^{(s)}). 
 \end{equation}
The training process can be terminated either after a prescribed total number of steps as given in Algorithm~\ref{alg:iteration}, or by the following stopping criterion: for some $\epsilon_E,\epsilon_V>0$,
\begin{equation}\label{stop}
\begin{aligned}
\left|E[\Psi_{\theta^{(t+1)}}]-E[\Psi_{\theta^{(t)}}]\right|
&<\epsilon_E, \qquad
\operatorname{Var}\!\left[E_L[\Psi_{\theta^{(t+1)}}]\right]
<\epsilon_V,
\quad t=0,1,\ldots
\end{aligned}
\end{equation}
 We will adopt the former manner in our experiments and use \eqref{stop} as a check. 
 To assess
the practical improvement from the TNN reduction, Table~\ref{tab:mem_trace} makes 
comparisons on the memory cost and computational time of the complete training steps that differ only in evaluation of
$\nabla^2f_{\rm TNN}$: the proposed way versus the direct one. 

\section{Numerical results}
\label{sec:baseline}
This section conducts numerical tests to demonstrate the effectiveness of BINN in solving \eqref{eq:rayleigh}. We shall  validate BINN through the theoretical predictions from Section~\ref{sec:asymptotics} as well as the ablation study. The transferability of BINN to cover different external potentials will be discussed in the end.

\subsection{Evaluation protocol}
\label{sec:setup}
We first establish the common evaluation protocol used across the section.
The experiments here fix batch size $B=256$ and use the same BINN architecture:
the network
${\varphi}_{\theta_1}=(\varphi_{1},\ldots,\varphi_{R})$
has hidden width $64$ and depth $2$, with $R=32$ outputs, while
$\varphi_{\theta_2}$ has hidden width $32$ and depth $1$.
All runs\footnote{Codes are available at \url{https://github.com/kongzhizhong/Boson-informed-Neural-Network}} are executed on a single NVIDIA GeForce RTX 5070~Ti GPU with
$16$~GB memory. 

As benchmarks, let $E_{\rm MF}$ and $E_{\rm DMC}$ denote the mean-field and
DMC reference energies~\cite{foulkes2001quantum}. 
Details of the two references are specified in 
Appendix~\ref{app:computational}. Let $E_{\rm BINN}$ be the energy obtained from BINN and \(\delta_{\rm DMC}\)
denote the relative residual energy gap between
$E_{\mathrm{BINN}}$ and its DMC projection:
\[
\delta_{\rm DMC}
=
\frac{|E_{\rm BINN}/N-E_{\rm DMC}/N|}
{|E_{\rm DMC}/N|}.
\] 
Additionally, we use the  coefficient of variation of the local energy to assess the convergence of training toward GS:
\[
\mathrm{CV}(E_L)
= \frac{1}{|\bar E_L|}
\left[\frac{1}{B-1}\sum_{s=1}^{B}
\bigl(E_L(\bm{R}^{(s)})-\bar{E}_L\bigr)^2\right]^{1/2},
\qquad
\bar{E}_L := \frac{1}{B}\sum_{s=1}^B E_L(\bm{R}^{(s)}).
\]
Recall the zero-variance criterion of $\operatorname{Var}_{\rho_\theta}[E_L]$ in Section~\ref{sec:vmc}, so
$\mathrm{CV}(E_L)$ which is the relative standard deviation of the local energy should approach zero. The gap to the corresponding asymptotic limit is defined separately for each benchmark regime in Section~\ref{sec:fixed_range}. In all tables, numbers in parentheses indicate statistical uncertainties in the last digit(s). For example, $1.017(4)$ means $1.017 \pm 0.004$. 

\begin{table}[t!]
\small
\caption{Fixed-range regime test with DMC validation for a 2D harmonic trap; here
$e_{\rm Hartree}(\sigma)=3.938$, ${\rm gap}:=e_{\rm Hartree}-E_{\rm BINN}/N$.
}
\label{tab:2d_harmonic}
\begin{center}
\renewcommand{\arraystretch}{1.2}
\begin{tabular}{crrrrrr}
\toprule
$N$ & $E_{\rm MF}/N$ & $E_{\rm BINN}/N$ & $E_{\rm DMC}/N$ & $\delta_{\rm DMC}$ & gap  & $\mathrm{CV}(E_L)$ \\
\midrule
 50 & $3.9015$ & $3.4377(1)$ & $3.4366(2)$ & $0.033\%$ & $0.5003$ &  $0.0131$ \\
 75 & $3.9137$ & $3.5912(1)$ & $3.5891(1)$ & $0.059\%$ & $0.3468$ &  $0.00980$ \\
100 & $3.9197$ & $3.6733(1)$ & $3.6723(1)$ & $0.028\%$ & $0.2647$ &  $0.00788$ \\
125 & $3.9234$ & $3.7239(1)$ & $3.7217(1)$ & $0.059\%$ & $0.2141$ &  $0.00662$ \\
150 & $3.9258$ & $3.7584(1)$ & $3.7581(1)$ & $0.009\%$ & $0.1796$ &  $0.00572$ \\
175 & $3.9275$ & $3.7833(1)$ & $3.7820(1)$ & $0.034\%$ & $0.1547$ &  $0.00505$ \\
200 & $3.9288$ & $3.8021(1)$ & $3.8023(1)$ & $0.005\%$ & $0.1359$ &  $0.00453$ \\
\bottomrule
\end{tabular}
\end{center}
\end{table}
\begin{figure}[t!]
\centering
\includegraphics[width=0.72\linewidth]{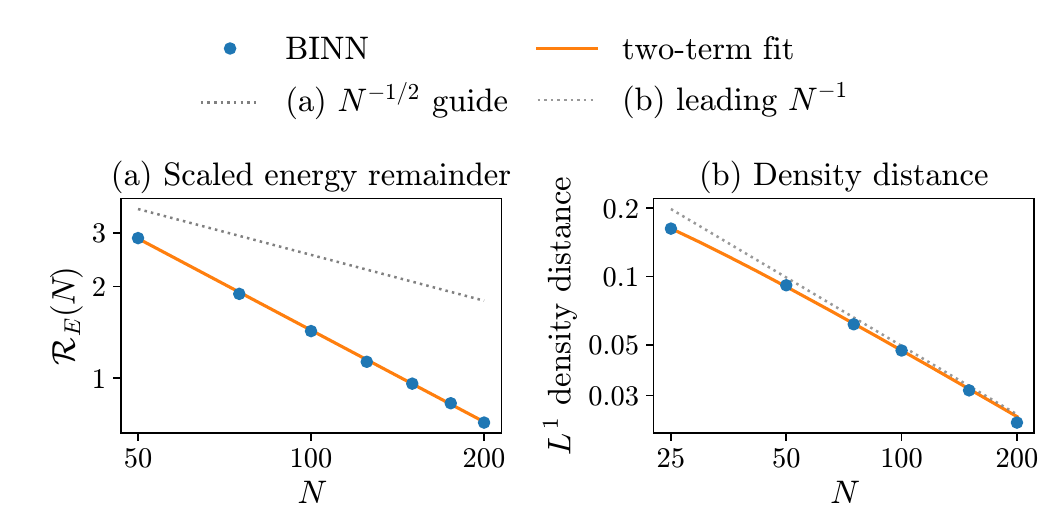}
\caption{Convergence to the large-$N$ Hartree limit for the fixed-range
2D harmonic trap, shown on log--log axes;
$e_{\rm Hartree}(\sigma)=3.938$.
\textbf{(a)} Scaled energy remainder
$\mathcal{R}_E(N)
:=N|e_{\rm Hartree}(\sigma)-E_{\rm BINN}/N-C_H/N|$, with $C_H=27.9$
fixed. Circular markers denote the BINN results, the solid curve is the
observed $N^{-1.00}$ power-law fit, and the dotted line is the
$N^{-1/2}$ guide predicted by \eqref{eq:Ecorr_trans}.
\textbf{(b)} $L^1$ distance (norm in theorems) between the BINN density $\rho_N$ and the
limiting Hartree density $\rho_{\rm H}$. The solid curve
shows the fit using $A_\rho/N+B_\rho/N^2$, and the dotted line its leading
term $A_\rho/N$; circular markers denote the BINN results.}
\label{fig:harm_2d_twopanel}
\end{figure}

\subsection{Accuracy and asymptotic benchmarks}
\label{sec:fixed_range}
We begin to validate BINN by comparing BINN energy with DMC, testing
convergence toward the Hartree and GP limits, and pushing the
shrinking-range calculation in 3D to \(N=350\).

\subsubsection{Fixed-range regime}
\label{sec:2d_harmonic} 
For the fixed-range case, we study the 2D isotropic
harmonic trap~\eqref{harmonic oscillator trapping potential} with
$\omega=1$, $\sigma=0.2$, $g=g_{\rm GP}/N$ and $g_{\rm GP}=100$. Table~\ref{tab:2d_harmonic}
shows that the relative energy gap between BINN and DMC is below $0.06\%$ for all $N$,
indicating that BINN reaches classical DMC-level accuracy. In addition, the local-energy coefficient of variation remains small across all system sizes, reflecting the reliability of our results. Meanwhile, the energy difference per-particle $e_{\rm Hartree}-E_{\rm BINN}/N$ (also the density errors in Fig.~\ref{fig:harm_2d_twopanel}(b)) exhibits an 
$\mathcal{O}(N^{-1})$ decay, consistent with the scaling
behavior \eqref{eq:Ecorr_expand}.   Moreover, we compute via~\cite{grech2013excitation,lewin2015bogoliubov} the coefficient  $C_H=27.9$ for~\eqref{eq:Ecorr_trans} and  Fig.~\ref{fig:harm_2d_twopanel}
tests the convergence. The observed remainder decays with a power-law exponent of
approximately \(-1.00\), faster than the expected $N^{-1/2}$ rate in \eqref{eq:Ecorr_trans}, indicating potential theoretical refinements. On the other hand, the  theoretical differences (e.g., $E_{\rm{MF}},\rm{gap}$ in Table~\ref{tab:2d_harmonic}) indicate limitations of the approximated models.

\subsubsection{Shrinking-range regime}
\label{sec:gp_limit}
Consider next the 3D shrinking-range regime as in 
Theorem~\ref{thm:lieb_seiringer}, which is more stringent due to the coupling of dimensionality and near-contact effects through $\sigma=c/N$. We fix $c=10$ and adopt the same setup as above but with \eqref{harmonic oscillator trapping potential} in 3D.

\begin{table}[t!]
\small
\setlength{\tabcolsep}{4pt}
\caption{Shrinking-range GP scaling $\sigma=c/N$ ($c=10$) on the 3D
harmonic trap; $e_{\rm GP}=2.506$, the rescaled
scattering length $N\,a(v_{\sigma})=4.919$ for every $N$,
$\mathrm{dis}_{L^1}=\|\rho_N-|\varphi_{\rm GP}|^2\|_{L^1}$. Wall (s) is the full 8000-step SR training time
and mem (GB) the peak GPU memory per training step.}
\label{tab:contact_exact}
\begin{center}
\renewcommand{\arraystretch}{1.2}
\begin{tabular}{rrrrrrrrr}
\toprule
$N$ & $\sigma$ & $E_{\rm BINN}/N$ & $E_{\rm DMC}/N$ & $\delta_{\rm DMC}$ & $\mathrm{dis}_{L^1}$ & $\mathrm{CV}(E_L)$ & wall (s) & mem (GB) \\
\midrule
 $50$ & $0.200$ & $2.56585(8)$ & $2.56561(5)$ & $0.009\%$ & $0.0474$ & $0.0103$ & $145$ & $0.36$ \\
 $60$ & $0.167$ & $2.55934(8)$ & $2.55901(5)$ & $0.013\%$ & $0.0431$ & $0.00914$ & $164$ & $0.50$ \\
 $70$ & $0.143$ & $2.55402(6)$ & $2.55353(4)$ & $0.019\%$ & $0.0378$ & $0.00830$ & $199$ & $0.65$ \\
 $80$ & $0.125$ & $2.54950(7)$ & $2.54900(4)$ & $0.020\%$ & $0.0352$ & $0.00767$ & $237$ & $0.82$ \\
 $90$ & $0.111$ & $2.54577(6)$ & $2.54574(4)$ & $0.001\%$ & $0.0317$ & $0.00720$ & $281$ & $1.01$ \\
$100$ & $0.100$ & $2.54274(5)$ & $2.54269(4)$ & $0.002\%$ & $0.0305$ & $0.00666$ & $332$ & $1.23$ \\
$150$ & $0.067$ & $2.53256(4)$ & $2.53234(2)$ & $0.009\%$ & $0.0207$ & $0.00521$ & $704$ & $2.60$ \\
\bottomrule
\end{tabular}
\end{center}
\end{table}

\begin{figure}[t!]
\centering
\includegraphics[width=0.72\linewidth]{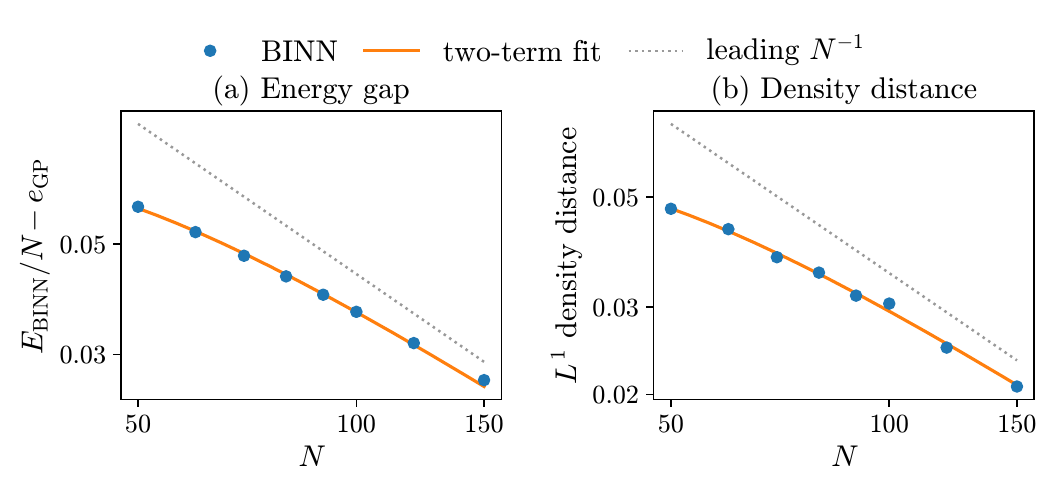}
\caption{Convergence to the large-$N$ GP limit for the 3D harmonic trap under shrinking-range scaling ($\sigma=10/N$, $N\le150$, log--log scale; $e_{\rm GP}=2.506$). \textbf{(a)} Per-particle energy excess $E_{\rm BINN}/N-e_{\rm GP}$. \textbf{(b)} $L^1$ distance between the BINN density $\rho_N$ and the limiting GP density $\vert{}\varphi_{\rm GP}\vert{}^2$. In both panels, circular markers denote BINN results, while solid and dotted curves represent the two-term fit ($D_1/N+b_G/N^2$ for panel a and $A_\rho/N+B_\rho/N^2$ for panel b) and its leading $\mathcal{O}(N^{-1})$ term, respectively.}
\label{fig:contact_exact}
\end{figure}

As shown in Table~\ref{tab:contact_exact}, the energies yielded by BINN and DMC agree within $0.02\%$ across the entire range from $N = 50$ to $150$. Crucially, BINN preserves DMC-level accuracy even as the system size grows and the interaction range shrinks, highlighting its robust performance in high-dimensional, near-contact regimes. With increasing $N$, both the per-particle energy gap and the density error $\Vert{}\rho_N - \vert{}\varphi_{\rm GP}\vert{}^2\Vert{}_{L^1}$ show an $\mathcal{O}(N^{-1})$ decay as plotted in Fig. \ref{fig:contact_exact}, consistent with Theorem~\ref{thm:lieb_seiringer} and the asymptotic result \eqref{eq:bogoliubov_scaling}. Furthermore, the wall-clock time and memory consumption reported in Table~\ref{tab:contact_exact} validate the practical $\mathcal{O}(N^2)$ scaling discussed in Section~\ref{sec:cost}.

\subsubsection{Stress test in 3D}
\label{sec:reach}

We further test robustness of BINN by pushing the above 3D shrinking-range
calculation to \(N=350\), where \(\sigma=c/N\approx0.029\). On this scale,
BINN gives \(E_{\rm BINN}/N=2.51879(2)\), in agreement with the DMC value
\(E_{\rm DMC}/N=2.51859(3)\), with a relative discrepancy
\(\delta_{\rm DMC}=0.008\%\) and the density error in \(L^1\) norm is \(0.012\). 
The method thus maintains DMC-level accuracy even when the configuration
dimension reaches \(3N=1050\) and the interaction range becomes very small. 

Optimization also remains stable at this scale, as shown by the training
dynamics in Fig.~\ref{fig:reach_traj}. Starting from a stiff initial state,
the per-particle energy reaches a stable plateau within about 300
stochastic-reconfiguration steps, while the local-energy coefficient of
variation decreases rapidly and stabilizes around $0.003$.

\begin{figure}[t!]
\centering
\includegraphics[width=\linewidth]{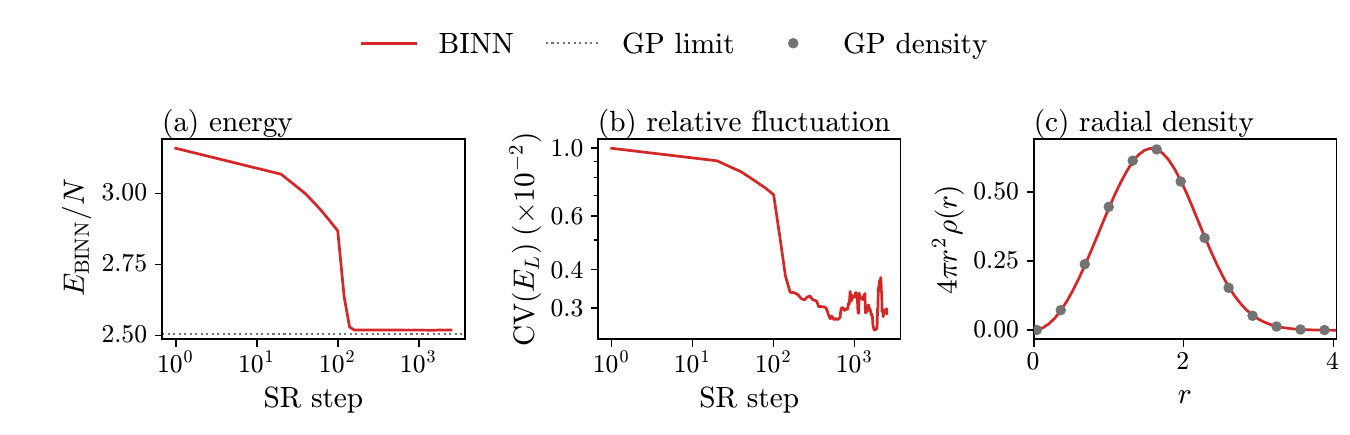}
\caption{BINN training and the converged density at the reach limit
$N=350$ ($\sigma\approx0.029$) for the 3D shrinking-range benchmark.
\textbf{(a)} Per-particle energy $E_{\rm BINN}/N$ versus SR step; the
gray dotted line marks the GP limit $e_{\rm GP}=2.506$.
\textbf{(b)} Local-energy coefficient of variation
$\mathrm{CV}(E_L)$ versus SR step.
\textbf{(c)} Radial mass distribution $4\pi r^2\rho(r)$ (here $\rho(r):=\frac{1}{4\pi r^2}
\int_{|\mathbf r|=r}\rho_N(\mathbf r)\,dS(\mathbf r)$) of the
converged BINN state (red curve), compared with the GP reference
(gray markers).}
\label{fig:reach_traj}
\end{figure}

\subsection{Ablation study}
\label{sec:ablation}

To further justify the architectural design of BINN, this section presents a series of ablation studies to validate the effectiveness of the cell Jastrow factor, the potential-aware envelope, and the TNN architecture. For this purpose, we consider a 2D problem with $N=100,\,\sigma = 0.2$ and an optical lattice external potential:
\[
v_{\rm ext}(\bm r)
=
V_0\left[\sin^2\!\left({\pi r_1}\right)+\sin^2\!\left({\pi r_2}\right)\right]
+\frac{1}{2}\lvert\bm r\rvert^2,
\qquad
\bm r=(r_1,r_2),\quad V_0=5.
\]


\textbf{Cell Jastrow factor.} To justify the introduction of $u_{\rm c}$ in \eqref{eq:jas_expression}, we compare the complete BINN (denoted as `Full') against a baseline lacking this analytical cell factor (denoted as `No-$u_{\rm c}$'). Furthermore, instead of using $\varphi_{\theta_2}(r^2)$, we evaluate an alternative configuration $u_{\rm c}+\varphi_{\theta_2}(r)$ (denoted as `MLP($r$)'). This validates the effectiveness of our squared-input feature.

\textbf{Potential-aware envelope.} To examine the necessity of the potential-aware envelope, we evaluate a variant of BINN (denoted as `No envelope') setting  $\chi^{\rm env}_{\theta_1}(\mathbf{r})\equiv0$ in \eqref{eq:sym_split}. Additionally, we compare with another variant (denoted as `No-$V_{\rm ext}$ term') that adopts simply
$\chi^{\rm env}_{\theta_1}(\mathbf{r})=-\alpha|\mathbf{r}|^2$ ($\alpha$ trainable) inspired by \cite{pfau2020ab} instead of \eqref{eq:envelope}. 

\begin{figure}[t!]
\centering
\includegraphics[width=\linewidth]{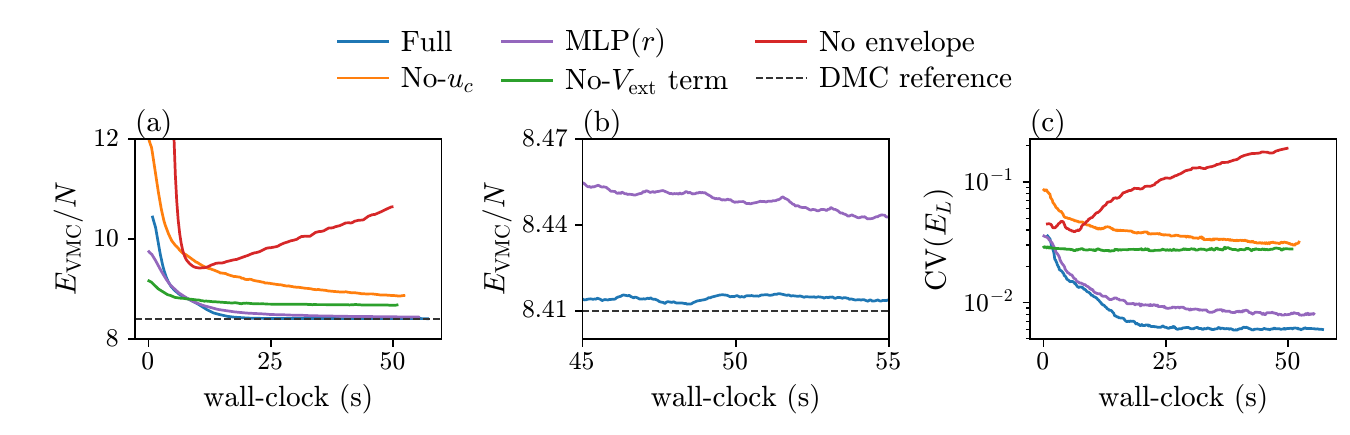}
\caption{Training dynamics of the BINN ablations on the 2D optical
lattice ($N=100$, $V_0=5$), plotted against wall-clock time.
\textbf{(a)} Per-particle VMC energies of all BINN variants.
\textbf{(b)} Late-time/zoom-in comparison of Full and MLP($r$).
\textbf{(c)} Local-energy coefficient of variation
$\mathrm{CV}(E_L)$, shown on a logarithmic scale. The Fully connected and
DeepSets trajectories lie outside the displayed ranges; their final
values are reported in Table~\ref{tab:ablation}.}
\label{fig:ablation_curves}
\end{figure}

\begin{table}[t!]
\small
\caption{Results of BINN ablations and generic neural baselines on the 2D optical lattice
($N=100$, $V_0=5$); here the Hartree and DMC reference energies are
\(E_{\rm MF}/N=8.657\) and \(E_{\rm DMC}/N=8.410\), respectively; $\Gamma_{\rm CV}=\mathrm{CV}_{\rm test}(E_L)/
\mathrm{CV}_{\rm Full}(E_L)$ compares the final local-energy coefficient
of variation of each variant with that of the complete BINN.}
\label{tab:ablation}
\begin{center}
\renewcommand{\arraystretch}{1.18}
\setlength{\tabcolsep}{3.5pt}
\footnotesize
\begin{tabular}{@{}lccccccc@{}}
\toprule
\multirow{2}{*}{Metric}
& \multirow{2}{*}{Full}
& \multirow{2}{*}{No-$u_{\rm c}$}
& \multirow{2}{*}{MLP($r$)}
& No-$V_{\rm ext}$
& No
& Fully
& \multirow{2}{*}{DeepSets} \\
& & & & term & envelope & connected & \\
\midrule
$E_{\rm VMC}/N$ & $\mathbf{8.415}$ & $8.866$ & $8.447$ & $8.682$ & $10.65$ & $51.0$ & $51.4$ \\
$\Gamma_{\rm CV}$ & $\mathbf{1.0}$ & $4.7$ & $1.3$ & $4.5$ & $31$ & $19$ & $19$ \\
\bottomrule
\end{tabular}
\end{center}
\end{table}

\textbf{Generic neural baselines.} To assess the effectiveness of the overall framework, we compare BINN with two baseline architectures: a fully connected network and DeepSets~\cite{pescia2022neural},  denoted by “Fully connected” and “DeepSets”, respectively. Neither baseline incorporates the Jastrow factor or the envelope. The results highlight the importance of our tailored structural design of BINN.

The complete BINN model reaches $E_{\rm BINN}/N=8.415$, close to the DMC
reference $E_{\rm DMC}/N=8.410$ and below the Hartree value
$E_{\rm MF}/N=8.657$. To compare local-energy fluctuations, we define
$
\Gamma_{\rm CV}
:=\mathrm{CV}_{\rm test}(E_L)/{\mathrm{CV}_{\rm Full}(E_L)},$
where $\mathrm{CV}_{\rm test}(E_L)$ is the final coefficient of variation
for the variant in the corresponding table column and
$\mathrm{CV}_{\rm Full}(E_L)$ is that of the complete BINN. For convenience, the VMC energies of all variants and neural baselines below are collectively denoted by \(E_{\rm VMC}\). As shown in Fig.~\ref{fig:ablation_curves} and Table~\ref{tab:ablation}, removing any of the key components of BINN leads to increase in both \(E_{\rm VMC}\) and \(\Gamma_{\rm CV}\), thereby degrading overall performance. The milder increase of MLP($r$) isolates the
additional regularity supplied by the squared radial input.

\subsection{Application to various trappings}
\label{sec:geometry_stress}
Last but not least, 
we apply BINN to several typical trapping scenarios in 2D. That is to set $v_{\mathrm{ext}}$ as the following four functions with fixed $N=50$, $g_{\rm GP}=100$ and  $\sigma=0.2$: 
\begin{enumerate}[label=(\alph*),leftmargin=*,itemsep=2pt]
\item Mexican-hat ring: $v_{\rm ext}(\bm{r})=-4|\bm{r}|^2+\tfrac12|\bm{r}|^4$;
\item Deep optical lattice: $v_{\rm ext}(\bm{r})=V_0\sum_{k=1}^{2}\sin^2({\pi r_k})+\tfrac12|\bm{r}|^2$ with $V_0=10$;
\item Double well: $v_{\rm ext}(\bm{r})=\tfrac12(r_1^2-x_0^2)^2+\tfrac12 r_2^2$ with $x_0=1.68$;
\item Smooth optical box: $v_{\rm ext}(\bm{r})=30\sum_{k=1}^{2}(r_k/L_{box})^{8}$ with $L_{box}=3$.
\end{enumerate}

\begin{figure}[t!]
\centering
\includegraphics[width=0.9\linewidth]{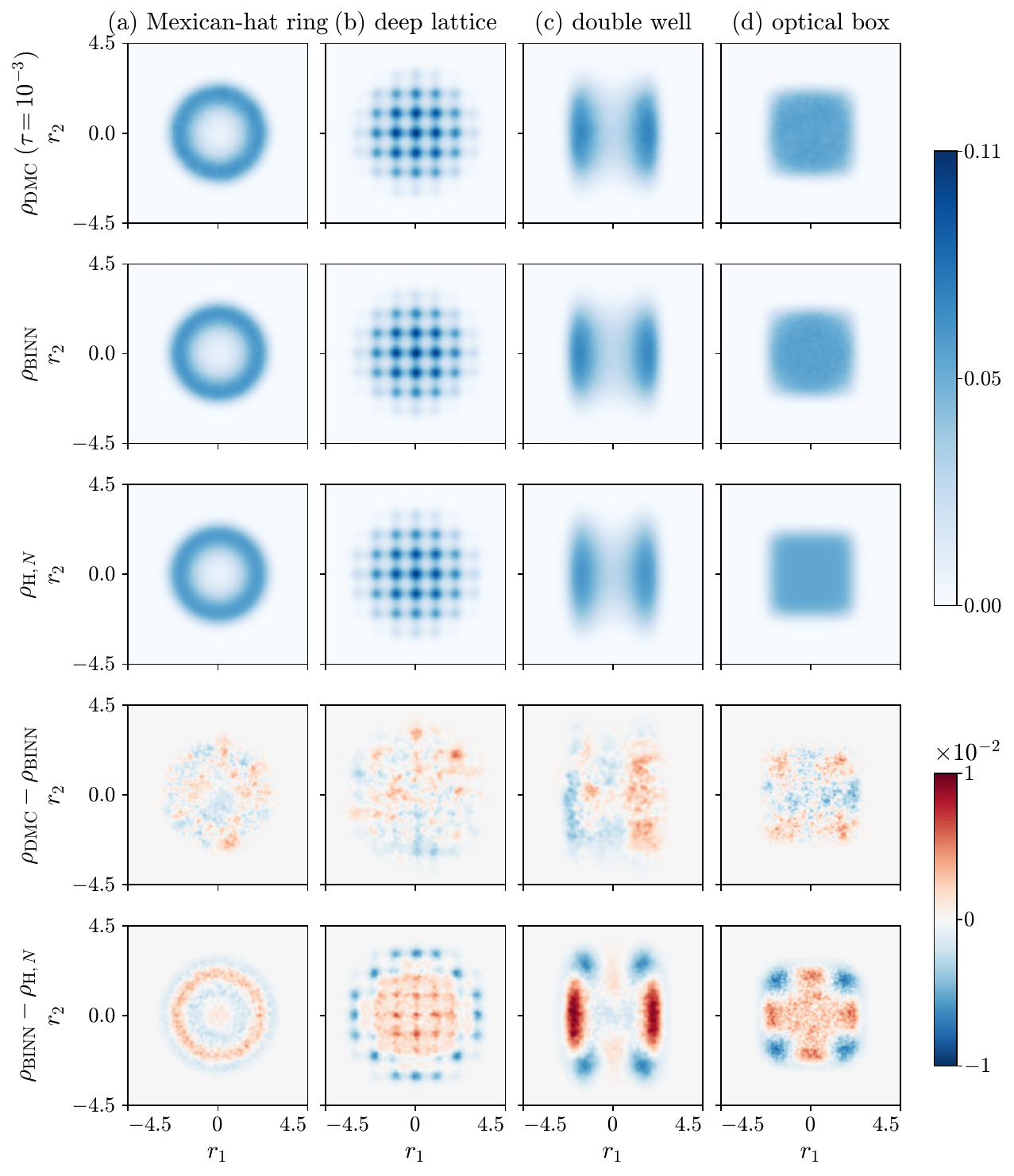}
\caption{Density function distributions for trapping potentials given in   
each column. The first two rows 
show the single-particle densities $\rho_{\rm DMC},\,\rho_{\rm BINN}$ computed from \eqref{single-partical density} with GS wave functions of DMC and BINN.
The third row shows the finite-$N$ Hartree density
$\rho_{{\rm H},N}:=|\varphi_{{\rm H},N}|^2$, where
$\varphi_{{\rm H},N}$ is a normalized minimizer of $E_{\rm {MF}}$ \eqref{E_MF def}. The last two rows show the differences $\rho_{\rm DMC}-\rho_{\rm BINN}$
and $\rho_{\rm BINN}-\rho_{{\rm H},N}$.}
\label{fig:geometry_combined}
\end{figure}

Across the four cases, the relative DMC discrepancy in energy \(\delta_{\rm DMC}=\)
\(0.24\%\), \(0.18\%\), \(0.59\%\), and \(0.38\%\) for the ring, lattice,
double-well and box traps, respectively. These sub-percent discrepancies illustrate the generalizability of BINN across different trapping potentials and again justify its uniform design that inherently encodes the trapping environment. Fig.~\ref{fig:geometry_combined} displays the 
spatial distributions of density functions of DMC, BINN and the finite-$N$ Hartree
density. The results further demonstrate the  capability of BINN in applications to different scenarios, where spatially macroscopic shapes as well as interaction-induced correlations are well captured. 

\section{Conclusion}
\label{sec:conclusion}
We introduced the Boson-informed neural network (BINN), a novel and robust framework designed for variational Monte Carlo computation of ground states for interacting bosons in continuous space. By explicitly decomposing the system's dominant physical scales through a Symmetry-Jastrow factorization, BINN effectively separates microscopic interaction stiffness via a cell Jastrow factor from macroscopic confinement encoded in a potential-aware envelope. This architectural design allows the neural network to capture the smoother many-body residual. Extensive numerical tests in both two and three dimensions demonstrate that BINN achieves state-of-the-art accuracy, approaching the Hartree and Gross-Pitaevskii limits. Notably, BINN exhibits exceptional performance in the challenging shrinking-range regime, maintaining Diffusion Monte Carlo level accuracy for particle numbers up to 350 in three dimensions, with a relative energy discrepancy as low as 0.008\%. Furthermore, ablation studies and applications across various trapping potentials—including ring-shaped, lattice, double-well, and box traps—underscore the necessity of our proposed architectural components and the robust generalizability of BINN to diverse physical scenarios.

\appendix
\section{Proof of Theorem~\ref{prop:cusp_laplacian}}
\label{app:proof_cusp_laplacian}

\begin{proof}
Throughout, $\nabla_{\rm rad}^2=\partial_r^2+\tfrac{d-1}{r}\partial_r$,
and $\mathcal{Q}_\sigma[u_{\rm c}]=\nu_\sigma$ on $[0,\ell]$
by~\eqref{eq:cell_defect}. Moreover $f_\sigma>0$ on $[0,\ell]$, since
the principal eigenfunction is nodeless and carries the normalization
$f_\sigma(\ell)=1$, as in the construction of
Section~\ref{sec:cusp_jastrow}.
As the principal Neumann eigenvalue of
$-\nabla_{\rm rad}^2+v_\sigma$ on $B_\ell$, $\nu_\sigma$ admits the
variational characterization
\[
  \nu_\sigma
  =\min_{0\neq f\in H^1(B_\ell)}
   \frac{\int_{B_\ell}\bigl(|\nabla f|^2+v_\sigma f^2\bigr)}
        {\int_{B_\ell}f^2} .
\]
Both terms in the numerator are nonnegative because $v_\sigma\ge0$, so
$\nu_\sigma\ge0$. Evaluating the quotient at the test function
$f\equiv1$, for which $\nabla f=0$, gives
\[
  \nu_\sigma\;\le\;\frac{1}{|B_\ell|}\int_{B_\ell}v_\sigma
  \;\le\;\frac{1}{|B_\ell|}\int_{\mathbb{R}^d}v_\sigma
  \;=\;\frac{g}{|B_\ell|},
\]
where the second inequality extends the integral of the nonnegative
$v_\sigma$ from $B_\ell$ to $\mathbb{R}^d$ and the last equality is the
normalization of the pair potential~\eqref{eq:pair_pot}.
Since $g=(g\sigma^{2-d})\,\sigma^{d-2}\le\kappa\max(1,\sigma_0^{d-2})$,
both $g$ and $\nu_\sigma$ are bounded uniformly in $\sigma$.

We next show that $f_\sigma$ is nondecreasing. In divergence form, the
cell equation reads
$(r^{d-1}f_\sigma')'=r^{d-1}(v_\sigma-\nu_\sigma)f_\sigma$, and
$h:=r^{d-1}f_\sigma'$ vanishes at both endpoints by the Neumann
conditions. Since $v_\sigma$ is strictly decreasing in $r$ and
$f_\sigma>0$, the derivative $h'$ changes sign at most once on
$(0,\ell)$, from nonnegative to nonpositive; a function vanishing at
both endpoints whose derivative follows this sign pattern is
nonnegative throughout. Hence, $f_\sigma'\ge0$ and
$0<f_\sigma\le f_\sigma(\ell)=1$ on $[0,\ell]$.

Monotonicity gives the slope bound directly. For $r\in(0,\ell]$,
\[
  0\;\le\;u_{\rm c}'(r)=\frac{h(r)}{r^{d-1}f_\sigma(r)},
  \qquad
  h(r)=\int_0^r s^{d-1}(v_\sigma-\nu_\sigma)f_\sigma\,ds
  \;\le\;f_\sigma(r)\int_0^r s^{d-1}v_\sigma\,ds ,
\]
where the last step drops the nonpositive $-\nu_\sigma$ contribution and
uses $f_\sigma(s)\le f_\sigma(r)$ for $s\le r$; dividing by
$r^{d-1}f_\sigma(r)$ leaves
\[
  u_{\rm c}'(r)\;\le\;\frac{1}{r^{d-1}}\int_0^r s^{d-1}v_\sigma(s)\,ds .
\]
For \(r\le\sigma\), using
\(v_\sigma(s)\le v_\sigma(0)=g\sigma^{-d}\pi^{-d/2}\) gives
\(
  u_{\rm c}'(r)\le \frac{v_\sigma(0)r}{d}
  \le \frac{\pi^{-d/2}\kappa}{d\,\sigma},
\)
where we used \(r\le\sigma\) and \(g\sigma^{2-d}\le\kappa\). For \(r\ge\sigma\), the normalization
\(\int_{\mathbb{R}^d}v_\sigma=g\) from~\eqref{eq:pair_pot}
can be written in polar coordinates as
\(g=|\mathbb{S}^{d-1}|\int_0^\infty s^{d-1}v_\sigma(s)\,ds\),
we obtain, for \(r\ge\sigma\),
\[
u_{\rm c}'(r)
\le r^{1-d}\int_0^r s^{d-1}v_\sigma(s)\,ds
\le \frac{g\,r^{1-d}}{|\mathbb{S}^{d-1}|}
\le \frac{\kappa}{|\mathbb{S}^{d-1}|\,r},
\]
where \( |\mathbb{S}^{d-1}| \) denotes the surface area of the unit
sphere and the last inequality follows from
\(g r^{2-d}\le g\sigma^{2-d}\le\kappa\) for \(d\ge2\).
since \(r\ge\sigma\) and \(d\ge2\) imply
\(g r^{2-d}\le g\sigma^{2-d}\le\kappa\).
If $r\le\sigma$, then $\sigma+r\le2\sigma$; if $r\ge\sigma$, then
$\sigma+r\le2r$. Hence, the two estimates combine to
$0\le u_{\rm c}'(r)\le C/(\sigma+r)$ on
$(0,\ell]$, with $C$ depending only on $d$ and $\kappa$.

With $u_{\rm NN}=\varphi_{\theta_2}(r^2)$ and we have
\(
  \nabla_{\rm rad}^2 u_{\rm NN}=2d\,\varphi_{\theta_2}'(r^2)+4r^2\varphi_{\theta_2}''(r^2).
\)
By the definition of $M_\Theta$, both expressions are bounded uniformly for
$\theta_2\in\Theta$ and $0\le r\le\ell$, and
$|u_{\rm NN}'(r)|\le 2M_\Theta r$. Expanding $\mathcal{Q}_\sigma$ at $u=u_{\rm c}+u_{\rm NN}$
and using the cell defect,
\[
  \mathcal{Q}_\sigma[u_{\rm c}+u_{\rm NN}]
  =\nu_\sigma-\nabla_{\rm rad}^2 u_{\rm NN}-2\,u_{\rm c}'\,u_{\rm NN}'-|u_{\rm NN}'|^2 .
\]
On $[0,\ell]$ the eigenvalue and network estimates bound $\nu_\sigma$,
$\nabla_{\rm rad}^2 u_{\rm NN}$, and $|u_{\rm NN}'|^2\le 4M_\Theta^2 r^2$,
while the slope bound controls the cross term,
\[
  |u_{\rm c}'\,u_{\rm NN}'|\le\frac{C}{\sigma+r}\,2M_\Theta r
  =2CM_\Theta\,\frac{r}{\sigma+r}\le 2CM_\Theta ,
\]
where the even feature $r^2$ supplies the factor $r$ that offsets the
$1/(\sigma+r)$ growth of $u_{\rm c}'$. Collecting the four terms by the
triangle inequality,
\[
  \sup_{0\le r\le\ell}
  \bigl|\mathcal{Q}_\sigma[u_{\rm c}+u_{\rm NN}](r)\bigr|
  \;\le\;
  \nu_\sigma
  +\sup_{0\le r\le\ell}\bigl|\nabla_{\rm rad}^2u_{\rm NN}\bigr|
  +2\sup_{0\le r\le\ell}\bigl|u_{\rm c}'\,u_{\rm NN}'\bigr|
  +\sup_{0\le r\le\ell}\bigl|u_{\rm NN}'\bigr|^2 ,
\]
and each supremum on the right has been bounded above by a constant
depending only on $M_\Theta$, $\ell$, $\kappa$, $\sigma_0$, and $d$, uniformly
in $0<\sigma\le\sigma_0$ and $\theta_2\in\Theta$; this is the stated bound
$C$.
\end{proof}

\section{Implementation details for experiments}
\label{app:computational}
We provide here the extra 
implementation details of our numerical experiments. %

\subsection{Network architecture and training protocol}
We count only hidden affine layers when reporting network depth; the output
layer is not included.  All hidden layers use the $\tanh$ activation, and all
output layers are affine.  Table~\ref{tab:binn_architecture} collects the
network architecture and training settings used in the reported BINN
experiments. The
MLP is shared over all particle indices $i$.
\begin{table}[h!]
\small
\centering
\caption{BINN architecture and training protocol; $a\to b\to c$ lists
successive layer widths.}
\label{tab:binn_architecture}
\begin{tabular}{lcc}
\toprule
Parameter or component & \multicolumn{2}{c}{Value} \\
\midrule
TNN log-orbital $\varphi_{\theta_1}$
  & \multicolumn{2}{c}{$d\to64\to64\to 32$} \\
neural Jastrow residual $\varphi_{\theta_2}$
  & \multicolumn{2}{c}{$1\to32\to1$}\\
cell radius $\ell$             & \multicolumn{2}{c}{$1.0$} \\
walkers $B$                  & \multicolumn{2}{c}{$256$} \\
MALA moves per SR step       & \multicolumn{2}{c}{$5$} \\
\midrule
                             & 2D fixed-range & 3D shrinking-range \\
\cmidrule(lr){2-3}
SR step size $\eta_t$        & $0.018/(1+t/400)$ & $0.016/(1+t/400)$ \\
training steps $t$              & $2500$            & $8000$ \\
\bottomrule
\end{tabular}
\end{table}


\subsection{DMC projection validation}
DMC is used only as a post-training projection check, with the converged
BINN state as the importance-sampling guide~\cite{foulkes2001quantum}; the
projection is continued to imaginary time $T_{DMC}=8$ with $\tau$  the
imaginary-time step and $E_{\rm DMC}(\tau)$ the corresponding energy. 
Tables~\ref{tab:2d_harmonic} and~\ref{tab:contact_exact} report $E_{\rm DMC}(0)$ as $E_{\rm DMC}$, where $E_{\rm DMC}(0)=E_{\rm DMC}(\tau)+\mathcal{O}(\tau)$ is given  by least squares using the values at
$\tau\in\{0.004,0.002,0.001\}$. Sections~\ref{sec:reach}
and~\ref{sec:geometry_stress} use only $\tau=0.001$ and $0.002$, giving
$E_{\rm DMC}(0)=2E_{\rm DMC}(0.001)-E_{\rm DMC}(0.002)$.

\subsection{Mean-field references}
The 2D Hartree and 3D contact-GP references in
Section~\ref{sec:fixed_range} are computed by an independent  imaginary-time Fourier-spectral evolution \cite{BaoCai}, used only for comparison.


\bibliographystyle{siamplain}
\bibliography{refs}

\end{document}